\documentclass[11pt,english,letterpaper]{article}
\date{}
\usepackage{babel}
\usepackage[T1]{fontenc}
\usepackage[latin9]{inputenc}
\usepackage{float}
\usepackage{amsmath}
\usepackage{graphicx}
\usepackage{amssymb}
\usepackage{verbatim}
\usepackage{fullpage}
\usepackage{geometry}
\usepackage{framed}
\usepackage{color}
\usepackage{theorem}

\floatstyle{ruled}
\newfloat{algorithm}{tbp}{loa}
\floatname{algorithm}{Algorithm}

\newtheorem{definition}{Definition}
\newtheorem{lemma}{Lemma}

\newtheorem{proposition}{Proposition}
\newtheorem{notation}{Notation}

\newenvironment{proof}{{\noindent\bf Proof. } }{{\hfill $\Box$}}

\begin{document}

\title{Dynamic FTSS in Asynchronous Systems: \\the Case of Unison\protect\footnote{This work was funded in part by ANR project SHAMAN, ALADDIN, and SPADES.} $^,$\footnote{A preliminary version of this work was published as a 2-pages brief announcement in DISC'09 \cite{DPT09ca}.}}

\author{
Swan Dubois$^{1,2}$ \and Maria Potop-Butucaru$^{1,2}$ \and S\'ebastien Tixeuil$^{1,3}$}

\maketitle

\vspace{-0.75cm}

\begin{center}
\scriptsize{
 $^{1}$ UPMC Sorbonne Universit\'es\\
 $^{2}$ INRIA Rocquencourt, Project-team REGAL\\
 $^{3}$ Institut Universitaire de France\\
 Postal adress: LIP6, Case 26-00/225, 4 place Jussieu, 75005 Paris (France)\\
 Mail: \{swan.dubois,maria.gradinariu,sebastien.tixeuil\}@lip6.fr -- Fax: +33 1 44 27 74 95 
}
\end{center}

\begin{abstract}
Distributed fault-tolerance can mask the effect of a limited number of permanent faults, while self-stabilization provides forward recovery after an arbitrary number of transient faults hit the system. FTSS (Fault-Tolerant Self-Stabilizing) protocols combine the best of both worlds since they tolerate simultaneously transient and (permanent) crash faults. To date, deterministic FTSS solutions either consider static (\emph{i.e.} fixed point) tasks, or assume synchronous scheduling of the system components. 

In this paper, we present the first study of deterministic FTSS solutions for dynamic tasks in asynchronous systems, considering the unison problem as a benchmark. Unison can be seen as a local clock synchronization problem as neighbors must maintain digital clocks at most one time unit away from each other, and increment their own clock value infinitely often. We present several impossibility results for this difficult problem and propose a FTSS solution (when the problem is solvable) for the state model that exhibits optimal fault containment.
\end{abstract}


\section{Introduction}\label{sec:Introduction}

The advent of ubiquitous large-scale distributed systems advocates that tolerance to various kinds of faults and hazards must be included from the very early design of such systems. \emph{Self-stabilization}~\cite{D74j,D00b} is a versatile technique that permits forward recovery from any kind of \emph{transient} fault, while \emph{Fault-tolerance}~\cite{FLP85j} is traditionally used to mask the effect of a limited number of \emph{permanent} faults. Making distributed systems tolerant to both transient and permanent faults is appealing yet proved difficult~\cite{AH93c,GP93c} as impossibility results are expected in many cases.

The seminal works of~\cite{AH93c,GP93c} define FTSS protocols as protocols that are both fault tolerant and self-stabilizing, \emph{i.e.} able to tolerate a few crash faults as well as arbitrary initial memory corruption. In~\cite{AH93c}, impossibility results for size computation and election in asynchronous systems are presented, while unique naming is proved possible. In~\cite{GP93c}, a general transformer is presented for synchronous systems, as well as positive results with failure detectors. The transformer of~\cite{GP93c} was later proved impossible to transpose to asynchronous systems due to the impossibility of tight synchronization in the FTSS context. For \emph{local} tasks (\emph{i.e.} tasks whose correctness can be checked locally, such as vertex coloring), the notion of \emph{strict} stabilization was proposed~\cite{NA02c,MT07j}. Strict stabilization guarantees that there exists a \emph{containment radius} outside which the effect of permanent faults is masked, provided that the problem specification makes it possible to break the causality chain that is caused by the faults. Strong stabilization~\cite{MT06cb,DMT10ca,DMT10cd} weakens this requirement and ensures processes outside the containment radius are only impacted a finite number of times by the Byzantine nodes.

It turns out that FTSS possibility results in fully \emph{asynchronous} systems known to date are restricted to \emph{static} tasks, \emph{i.e.} tasks that require eventual convergence to some global fixed point (tasks such as naming or vertex coloring fall in this category). In this paper, we consider the more challenging problem of \emph{dynamic} tasks, \emph{i.e.} tasks that require both eventual safety and liveness properties (examples of such tasks are clock synchronization and token passing). Due to the aforementioned impossibility of tight clock synchronization, we consider the \emph{unison} problem, which can be seen as a \emph{local} clock synchronization problem. In the unison problem~\cite{M91j}, each node is expected to keep its digital clock value within one time unit of every of its neighbors' clock values (weak synchronization), and increment its clock value infinitely often (liveness). Note that in synchronous systems where the underlying topology is a fully connected graph in which clocks have discrete time unit values, unison induces tight clock synchronization. Several self-stabilizing solutions exist for this problem~\cite{BPV04c,BPV05c,CFG92c,GH90j}, both in synchronous and asynchronous systems, yet none of those can tolerate crash faults.

As a matter of fact, there exists a number of FTSS results for \emph{dynamic} tasks in \emph{synchronous} systems. \cite{DW97j,PT97j} provide self-stabilizing clock synchronization that is also~\emph{wait free}, \emph{i.e} that tolerate napping faults, in complete networks. Also \cite{D97j} presents a FTSS clock synchronization for general networks. Still in synchronous systems, it was proved that even \emph{malicious} (\emph{i.e.} Byzantine) faults can be tolerated, to some extent. In~\cite{BDH08c,DW04j}, probabilistic FTSS protocols were proposed for up to one third of Byzantine processors, while in \cite{DH07cb,HDD06c} deterministic solution tolerate up to one fourth and one third of Byzantine processors, respectively. Note that all solutions presented in this paragraph are for fully \emph{synchronous} systems. \cite{HBD10c} is a notable exception since it proposes a probabilistic solution to a clock synchronization problem in an asynchronous system.  

In this paper, we tackle the open issue of FTSS \emph{deterministic} solutions to \emph{dynamic} tasks in \emph{asynchronous} systems, using the unison problem as a case study. Our first negative results show that whenever two or more crash faults may occur, FTSS unison is impossible in any asynchronous setting. The remaining case of one crash fault drives the most interesting results (see Section~\ref{sec:negative}). The first main contribution of the paper is the characterization of two key properties satisfied by all previous self-stabilizing asynchronous unison protocols: \emph{minimality} and \emph{priority}. Minimality means that nodes maintain no extra variables but the digital clock value. Priority means that if incrementing the clock value does not break the local safety predicate between neighbors, then the clock value is actually incremented in a finite number of activations, even if no neighbor modifies its clock value. Then, depending on the fairness properties of the scheduling of nodes, we provide various results with respect to the possibility or impossibility of unison. When the scheduling is \emph{unfair} (only global progress is guaranteed), \emph{universal} FTSS unison (\emph{i.e.} unison that can operate on every graph of a particular class) is impossible. When the scheduling is \emph{weakly fair} (a processor that is continuously enabled is eventually activated), then it is impossible to solve universal FTSS unison by a protocol that satisfies either minimality or priority. The case of \emph{strongly fair} scheduling (a processor that is enabled infinitely often is eventually activated) is similar whenever the maximum degree of the graph is at least three. Our negative results still apply when the clock variable is unbounded, the local synchronization constraint is relaxed, and the scheduling is central (\emph{i.e.} a single processor is activated at any time).

On the positive side (Section~\ref{sec:positive}), we present a universal FTSS protocol for connected networks of maximum degree at most two (\emph{i.e.} rings and chains), which satisfies both minimality and priority properties. This protocol makes minimal system hypothesis with respect to the aforementioned impossibility results (maximum degree, fairness of the scheduling, etc.) and is optimal with respect to the containment radius that is achieved (\emph{no} correct processor is \emph{ever} prevented from incrementing its clock). This protocol assumes that the scheduling is central. Table~\ref{table1} provides a summary of the main results of the paper. Remaining open questions are discussed in Section~\ref{sec:conclusion}.

\begin{table}[!h]
\centering
	\begin{tabular}{|c||c|c|c|c|c|c|}
	\cline{2-7}
	\multicolumn{1}{c||}{}  & Unfair & \multicolumn{2}{c|}{Weakly fair} & \multicolumn{3}{c|}{Strongly fair}\tabularnewline
	\cline{3-7}
	 \multicolumn{1}{c||}{} &  & Minimal & Priority & \multicolumn{2}{c|}{$\Delta\geq3$} & $\Delta\leq2$ \tabularnewline
	\cline{5-6} 
	\multicolumn{1}{c||}{}  &  & & & Minimal & Priority & \tabularnewline
	\hline
	\hline
	 $f=1$ & Impossible & Impossible & Impossible & Impossible & Impossible & Possible \tabularnewline
	       & (Prop. \ref{prop:impUF}) & (Prop. \ref{prop:impWFMin}) & (Prop. \ref{prop:impWFPri}) & (Prop. \ref{prop:impSFMin}) & (Prop. \ref{prop:impSFPri}) & (Prop. \ref{prop:ftss})\tabularnewline
	\cline{1-7}
	  $f\geq2$ & \multicolumn{6}{c|}{Impossible (Prop. \ref{prop:impf2})}\tabularnewline
	\hline
	\end{tabular}
\caption{Summary of results}
\label{table1}
\end{table}

\section{Model, Problem and Specifications}\label{sec:Model}

We model the network as an undirected connected graph $G=(V,E)$ where $V$ is a set of processors and $E$ is a binary relation that denotes the ability for two processors to communicate ($(p,q)\in E$ if and only if $p$ and $q$ are \emph{neighbors}). We consider only \emph{anonymous} systems (\emph{i.e.} there exists no unique identifiers for each processor) but we assume that every processor $p$ can distinguish its neighbors and locally label them. Each processor $p$ maintains $N_{p}$, the set of its neighbors' local labels. In the following, $n$ denotes the number of processors, and $\Delta$ the maximal degree. If $p$ and $q$ are two processors of the network, we denote by $d(p,q)$ the length of the shortest path between $p$ and $q$ (\emph{i.e} the \emph{distance} from $p$ to $q$). In this paper, we assume that the network can be hit by \emph{crash faults}, \emph{i.e.} some processors can stop executing their actions permanently and without any warning to their neighborhood. Since the system is assumed to be fully asynchronous, no processor can detect if one of its neighbors is crashed or slow.	

We consider the classical local shared memory model of computation (see \cite{D00b}) where communications between neighbors are modeled by direct reading of variables instead of exchange of messages. In this model, the program of every processor consists of a set of shared variables (henceforth, referred to as \emph{variables}) and a finite set of \emph{rules}. A processor can write to its own variables only, and read its own variables and those of its neighbors. Each rule consists of: $<$\emph{label}$>$::$<$\emph{guard}$>\longrightarrow<$\emph{statement}$>$. The label of a rule is simply a name to refer the action in the text. The guard of a rule in the program of $p$ is a Boolean predicate involving variables of $p$ and its neighbors. The statement of a rule of $p$ updates one or more variables of $p$. A statement can be executed only if the corresponding guard is satisfied (\emph{i.e.} it evaluates to true). The processor rule is then \emph{enabled}, and processor $p$ is \emph{enabled} in $\gamma\in \Gamma$ if and only if at least one rule is enabled for $p$ in $\gamma$. The state of a processor is defined by the current value of its variables. The state of a system (\emph{a.k.a.} the \emph{configuration}) is the product of the states of all processors. We also refer to the state of a processor and its neighborhood as a \emph{local configuration}. We note $\Gamma$ the set of all configurations of the system. 

A \emph{step} $\gamma\rightarrow\gamma'$ is defined as an atomic execution of a non-empty subset of enabled rules in $\gamma$ that transitions the system from $\gamma$ to $\gamma'$. An execution of a protocol $\mathcal{P}$ is a maximal sequence of configurations $\epsilon=\gamma_{0}\gamma_{1}\ldots\gamma_{i}\gamma_{i+1}\ldots$ such that, $\forall i\geq0,\gamma_{i}\rightarrow\gamma_{i+1}$ is a step if $\gamma_{i+1}$ exists (else $\gamma_{i}$ is a \emph{terminal} configuration). \emph{Maximality} means that the sequence is either finite (and no action of $\mathcal{P}$ is enabled in the terminal configuration) or infinite. $\mathcal{E}$ is the set of all possible executions of $\mathcal{P}$. A processor $p$ is \emph{neutralized} in step $\gamma_{i}\rightarrow\gamma_{i+1}$ if $p$ is enabled in $\gamma_{i}$ and is \emph{not} enabled in $\gamma_{i+1}$, yet did not execute any rule in step $\gamma_{i}\rightarrow\gamma_{i+1}$.

A \emph{scheduler} (also called \emph{daemon}) is a predicate over the executions. Recall that, in any execution, each step $\gamma \longrightarrow \gamma'$ results from a \emph{non-empty} subset of enabled processors \emph{atomically executing} a rule. This subset is chosen by the scheduler. A scheduler is \emph{central} if it chooses \emph{exactly one} enabled processor in any particular step, it is \emph{distributed} if it chooses \emph{at least one} enabled processor, and \emph{locally central} if it chooses \emph{at least one} enabled processor yet ensures that no two neighboring processors are chosen concurrently. A scheduler is \emph{synchronous} if it chooses \emph{every} enabled processor in every step. A scheduler is \emph{asynchronous} if it is either central, distributed or locally central. A scheduler may also have some \emph{fairness} properties. A scheduler is \emph{strongly fair} (the strongest fairness assumption for asynchronous schedulers) if every processor that is enabled \emph{infinitely often} is eventually chosen to execute a rule. A scheduler is \emph{weakly fair} if every \emph{continuously} enabled processor is eventually chosen to execute a rule. Finally, the \emph{unfair} scheduler has the weakest fairness assumption: it only guarantees that at least one enabled processor is eventually chosen to execute a rule. As the strongly fair scheduler is the strongest fairness assumption, any problem that cannot be solved under this assumption cannot be solved for all weaker fairness assumptions. In contrast, any algorithm performing under the unfair scheduler also works for all stronger fairness assumptions. 

\paragraph{Fault-containment and Stabilization.} In a particular execution $\epsilon$, we distinguish the set of processors $V^*$ that never crash in $\epsilon$ (\emph{i.e.} the set of \emph{correct} processors). By extension, for any part $C\subset V$, the set of correct processors in $C$ is denoted by $C^{*}$. As crashed processors cannot be distinguished from slow ones by their neighbors, we assume that variables of crashed processors are always readable. 

Let $\mathcal{P}$ be a problem to solve. A \emph{specification} of $\mathcal{P}$ is a predicate that is satisfied by every algorithm solving the problem. We recall definitions about stabilization and fault-tolerance.

\begin{definition} [self-stabilization \cite{D74j}] \label{def:self}
Let $\mathcal{P}$ be a problem, and $\mathcal{\mathcal{S}_{P}}$ a specification of $\mathcal{P}$. An algorithm $\mathcal{A}$ is self-stabilizing for $\mathcal{S_{P}}$ if and only if for every configuration  $\gamma_{0}\in\Gamma$, for every execution $\epsilon=\gamma_{0}\gamma_{1}\ldots$, there exists a finite prefix $\gamma_{0}\gamma_{1}\ldots\gamma_{l}$ of $\epsilon$ such that all executions starting from $\gamma_{l}$ satisfy $\mathcal{S_{P}}$.
\end{definition}

\begin{definition} [$(f,r)-$containment \cite{NA02c}]
Let $\mathcal{P}$ be a problem, and $\mathcal{\mathcal{S}_{P}}$ a specification of $\mathcal{P}$. A configuration $\gamma\in\Gamma$ is $(f,r)-$contained for specification $\mathcal{\mathcal{S}_{P}}$ if and only if, given at most $f$ crashed processors, every execution starting from $\gamma$, always satisfies $\mathcal{\mathcal{S}_{P}}$ on the sub-graph induced by processors that are at distance $r$ or more from any crashed processor.
\end{definition}

\begin{definition} [fault-tolerant self-stabilization (FTSS) \cite{AH93c,GP93c}] \label{def:ftss}
Let $\mathcal{P}$ be a problem, $ $ and $\mathcal{\mathcal{S}_{P}}$ a specification of $\mathcal{P}$. An algorithm $\mathcal{A}$ is fault-tolerant and self-stabilizing  with radius $r$ for $f$ crashed processors (and denoted by $(f,r)-FTSS$) for specification $\mathcal{\mathcal{S}_{P}}$ if and only if, given at most $f$ crashed processors, for every configuration  $\gamma_{0}\in\Gamma$, for every execution $\epsilon=\gamma_{0}\gamma_{1}\ldots$, there exists a finite prefix $\gamma_{0}\gamma_{1}\ldots\gamma_{l}$ of $\epsilon$ such that $\gamma_{l}$ is $(f,r)-$contained for specification $\mathcal{\mathcal{S}_{P}}$.
\end{definition}

\paragraph{Unison.} In the following, $c_{p}$ is the variable of processor $p$ that represents its clock value. Values are taken in the set of natural integers (that is, the number of states is unbounded, and a total order can be defined on clock values). Note that we do not consider the case of bounded clocks in this paper. We now define two notions related to local clock synchronization: the first one restricts the safety property to correct processors, while the second one considers all processors. We call \emph{drift} between two processors $p$ and $q$  the absolute value of the difference between their clock values. In this paper, we deal with unison that is a weak clock synchronization: we must ensure that clocks are eventually "close" from each other. More precisely, two processors $p$ and $q$ are \emph{in unison} if the drift between them is no more than $1$. We say that a configuration of the system is \emph{weakly synchronized} if any correct processor is in unison with its correct neighbors. More formally, 

\begin{definition}[weakly synchronized configuration]
Let $\gamma \in \Gamma$. We say that $\gamma$ is weakly synchronized, denoted by $\gamma \in \Gamma_{1}^{*}$, if and only if :	$\forall p \in V^{*}\,\forall q \in N_{p}^{*}\,\,|c_{p}-c_{q}|\leq 1$.
\end{definition}
	
We say that a configuration of the system is \emph{uniformly weakly synchronized} if any processor is in unison with all its neighbors (even with crashed ones). More formally,
 
\begin{definition}[uniformly weakly synchronized configuration]
Let $\gamma \in \Gamma$. We say that $\gamma$ is uniformly weakly synchronized, denoted by $\gamma \in \Gamma_{1}$, if and only if : $\forall p \in V,\forall q \in N_{p},\,|c_{p}-c_{q}|\leq 1$.
\end{definition}
	
Figure \ref{fig:synchro} gives some examples of weakly synchronized configurations.

\input{figure1}

We now specify the two variants of our problem (depending whether safety property is extended to crashed processors or not). Intuitively, asynchronous unison (respectively uniform asynchronous unison) ensures that the system is eventually (and remains forever) in a weakly (respectively uniformly weakly) synchronized configuration (safety property) and that clocks of correct processors are infinitely often incremented by $1$ (liveness condition). More formally,

\begin{definition}[asynchronous unison]
Let $\gamma_{0} \in \Gamma$. An execution $\epsilon=\gamma_{0}\gamma_{1}...$ is a legitimate execution for asynchronous unison, denoted by \textbf{AU}, if and only if:\\
\indent \textbf{Safety:} $\forall i\in\mathbb{N},\gamma_{i}\in\Gamma_{1}^{*}$.\\
\indent \textbf{Liveness:} Each processor $p\in V^{*}$ increments its clock (by $1$) infinitely often in $\epsilon$.
\end{definition}
	
\begin{definition}[uniform asynchronous unison]
Let $\gamma_{0} \in \Gamma$. An execution $\epsilon=\gamma_{0}\gamma_{1}...$ is a legitimate execution for uniform asynchronous unison, denoted by \textbf{UAU}, if and only if:\\
\indent \textbf{Safety:} $\forall i\in\mathbb{N},\gamma_{i}\in\Gamma_{1}$.\\
\indent \textbf{Liveness:} Each processor $p\in V^{*}$ increments its clock (by $1$) infinitely often in $\epsilon$.
\end{definition}

Note that an algorithm that complies to specification of \textbf{UAU} also complies to that of \textbf{AU} (the converse is not true) since $\Gamma_{1}\subseteq\Gamma_{1}^{*}$ (if no processor is crashed, we have: $\Gamma_{1}=\Gamma_{1}^{*}$, but if at least one processor is crashed, we have: $\Gamma_{1}\subsetneq \Gamma_{1}^{*}$). Note also that these two specifications do not forbid \emph{decrementing} clocks. Our specification generalizes the classical unison specification~\cite{CFG92c} as any solution to the former is also a solution of ours. Unison protocols that are useful in a distributed setting are those that do not know the underlying communication graph. We refer to \emph{universal} protocols to denote the fact that a protocol that can perform on every communication graph that matches a particular predicate (\emph{e.g.} every graph of degree less than two). To disprove universality of a protocol, it is thus sufficient to exhibit a particular communication graph in its acceptance predicate such that at least one possible execution does not satisfy the specification.

We now present two key properties satisfied by all known self-stabilizing unison protocols. Those properties are used in the impossibility results presented in Section~\ref{sec:negative}. We called these properties respectively \emph{minimality} and \emph{priority}.

Minimality means that nodes maintain no extra variables but the digital clock value. This implies that the code of a minimal unison can only refer to clocks or to predefined constants. We now state the formal definition of this property.

\begin{definition}[minimality]
A unison is \emph{minimal} if and only if every processor only maintains a clock variable.
\end{definition}

Priority means that if, for a given processor, incrementing the clock value does not break the local safety predicate with its neighbors, then its clock value is actually incremented in a finite number of activations, even if no neighbor modifies its clock value. This property implies that, if a processor can increment its clock without breaking unison with its neighbors, then it does so in finite time whether its neighbors are crashed or not. This property is similar to obstruction-freedom in the sense that the protocol only has very weak constraints about progress. We formally state this property in the following definition.

\begin{definition}[priority]
A unison is \emph{priority} if and only if it satisfies the following property: if there exists a processor $p$ such that $\forall q\in N_{p},(c_{q}=c_{p}$ or $c_{q}=c_{p}+1)$ in a configuration $\gamma_{i}$, then there exists a fragment of execution $\epsilon=\gamma_{i}...\gamma_{i+k}$ such that:\\
\indent - only $p$ is chosen by the scheduler during $\epsilon$.\\
\indent - $c_{p}$ is not modified during $\gamma_{i+j}\longrightarrow\gamma_{i+j+1}$, for $j\in\{0,...,k-2\}$.\\
\indent - $c_{p}$ is incremented during $\gamma_{i+k-1}\longrightarrow\gamma_{i+k}$.
\end{definition}
	
For example, protocols proposed by \cite{BPV04c,BPV05c,CFG92c,GH90j} fall in the category of minimal and priority unison using these definitions. Another example is the protocol of \cite{PT97j} that is priority but not minimal. To our knowledge, any existing unison protocol satisfies either minimality or priority.

\section{Impossibility Results}\label{sec:negative}

In this section we present a broad class of impossibility results related to the FTSS unison. First, we show a preliminary result that states that a processor cannot modify its clock value if it has two neighbors $q$ and $q'$ with $c_{q}=c_{p}-1$ and $c_{q'}=c_{p}+1$ (Lemma \ref{lem:blocage}). This property is further used in the sequel of this section. Proposition \ref{prop:impf2} proves that there exists no $(f,r)-$FTSS algorithm for any $r$ value if $f\geq2$. Furthermore, in  Proposition \ref{prop:impUF}, we prove that there exists no $(1,r)-$FTSS algorithm for \textbf{AU} under an unfair daemon for any $r$ value. Then we study the minimal and priority asynchronous unison and  prove there exists no $(1,r)-$FTSS algorithm for minimal or priority \textbf{AU} under a weakly fair daemon for any $r$ value (Lemma \ref{lem:impWFMin}, Propositions \ref{prop:impWFMin} and \ref{prop:impWFPri}). Finally, we prove there exists no $(1,r)-$FTSS algorithm for minimal or priority \textbf{AU} under a strongly fair daemon for any $r$ value if the network has a maximal degree of at least 3 (Lemma \ref{lem:impSFMin}, Propositions \ref{prop:impSFMin} and \ref{prop:impSFPri}). In the following we assume, for the sake of generality, the most constrained scheduler (the central one).

\subsection{Preliminaries}

First, we introduce a preliminary result that shows that in any execution of a universal $(f,r)-$ftss algorithm for \textbf{AU} (under an asynchronous daemon) a processor cannot modify its clock value if it has two neighbors $q$ and $q'$ such that: $c_{q}=c_{p}-1$ and $c_{q'}=c_{p}+1$.  

\begin{lemma}\label{lem:blocage}
Let $\mathcal{A}$ be a universal $(f,r)-$ftss algorithm for \textbf{AU} (under an asynchronous daemon). Let $\gamma$ be a configuration where a processor $p$ (such that $c_{p}\geq 1$) has two neighbors $q$ and $q'$ such that: $c_{q}=c_{p}-1$ and $c_{q'}=c_{p}+1$. If $p$ executes an action of $\mathcal{A}$ during the step $\gamma\longrightarrow\gamma'$, then this action does not modify the value of $c_{p}$.   If $\mathcal{A}$ is also minimal, then the processor $p$ is not enabled for $\mathcal{A}$ in $\gamma$.
\end{lemma}

\begin{proof}
Let $\mathcal{A}$ be a universal $(f,r)-$ftss algorithm for \textbf{AU} (under an asynchronous daemon). Let $G$ be a network and $\gamma$ be a configuration of $G$ such that no processor is crashed, $\gamma\in\Gamma_{1}$ and there exists a processor $p$ (such that $c_{p}\geq 1$) that has two neighbors $q$ and $q'$ such that: $c_{q}=c_{p}-1$ and $c_{q'}=c_{p}+1$.

Assume $p$ executes an action of $\mathcal{A}$ during the step $\gamma\longrightarrow\gamma'$ (and only $p$) such that this action modifies the value of $c_{p}$. Note that $c_{q}$ and $c_{q'}$ are identical in $\gamma$ and $\gamma'$. Let $\alpha$ be the value of $c_{p}$ in $\gamma$ and $\alpha'$ be the value of $c_{p}$ in $\gamma'$. Values of $\alpha$ and $\alpha'$ satisfy one of the two following relations:

\begin{description}
\item[Case 1:] $\alpha<\alpha'$.\\
This implies that $|\alpha'-c_{q}|=|\alpha'-\alpha|+|\alpha-c_{q}|>1$ (since $|\alpha'-\alpha|\geq 1$ by hypothesis and $|\alpha-c_{q}|=1$).
\item[Case 2:] $\alpha'<\alpha$.\\
This implies that $|\alpha'-c_{q'}|=|\alpha'-\alpha|+|\alpha-c_{q'}|>1$ (since $|\alpha'-\alpha|\geq 1$ by hypothesis and $|\alpha-c_{q'}|=1$).
\end{description} 

In the two above cases, $\gamma'\notin\Gamma_{1}$, hence the safety property of $\mathcal{A}$ is not satisfied.

If $\mathcal{A}$ is also minimal, then the previous result implies that  $p$ is not enabled for $\mathcal{A}$ in $\gamma$.
\end{proof}

\subsection{Impossibility Result due to the Number of Crashed Processors}

\begin{proposition}\label{prop:impf2}
For any natural number $r$, there exists no universal $(f,r)-$ftss algorithm for \textbf{AU} under an asynchronous daemon if $f\geq 2$.
\end{proposition}

\begin{proof}
Let $r$ be a natural number. Let $\mathcal{A}$ be a universal $(2,r)-$ftss algorithm for \textbf{AU} (under an asynchronous daemon). Consider a network represented by the following graph: $G=(V,E)$ with $V=\{p_{0},\ldots,p_{2(r+1)}\}$ and $E=\left\{\{p_{i},p_{i+1}\}|i\in\{0,\ldots,2r+1\}\right\}$.	Let $\gamma$ be the following configuration of the network: $p_{0}$ and $p_{2(r+1)}$ are crashed and $\forall i\in\{0,\ldots,2(r+1)\},c_{p_{i}}=i$ 		(all the other variables may have any value).

By Lemma \ref{lem:blocage}, no processor between $p_{2}$ and $p_{2r+1}$ can change its clock value in every execution starting from $\gamma$. This contradicts the definition of $\mathcal{A}$. Indeed, $p_{r+1}$ must eventually satisfy the specification of \textbf{AU} since the closest crashed processor is at $r$ hops away. In particular, any execution starting from $\gamma$ must contain a suffix where the clock of $p_{r+1}$ is infinitely often incremented. This contradiction shows us the result.
\end{proof}

\subsection{Impossibility Result due to Unfair Daemon}

\begin{proposition}\label{prop:impUF}
For any natural number $r$, there exists no universal $(1,r)-$ftss algorithm for \textbf{AU} under an unfair daemon.
\end{proposition}

\begin{proof}
Let $r$ be a natural number. Assume that there exists a universal $(1,r)-$ftss algorithm $\mathcal{A}$ for \textbf{AU} under an unfair daemon. Consider a network $G$, of diameter greater than $2r+2$ (note that in this case, at least one processor must eventually satisfy the specification of the \textbf{AU} problem). Let $p$ be a processor of $G$. Since the daemon is unfair, it can choose to never activate $p$ in an execution $\epsilon$ unless this processor becomes the only enabled processor of $G$ in a configuration of $\epsilon$ by definition.
			
For the sake of contradiction, assume that there exists a configuration $\gamma$ such that no processor is crashed and where $p$ is the only enabled processor of the network. Denote by $\gamma'$ the same configuration when $p$ is crashed. Note that the set of enabled processors is identical in $\gamma$ and $\gamma'$ by construction. As we assumed that only $p$ is enabled in $\gamma$, this implies that no correct processor is enabled in $\gamma'$. Hence, the system is deadlocked in $\gamma'$ and the specification of \textbf{AU} is not satisfied since no clock of correct processor can be updated. This contradiction implies that, for any configuration where no processor is crashed, at least two processors are enabled. 

Since there exists no configuration where $p$ is the unique enabled processor (in every execution starting from an arbitrary configuration), the unfair daemon can starve $p$ infinitely (if no crash occurs). This contradicts the liveness property of $\mathcal{A}$ since $p$ cannot update its clock in this execution.
\end{proof}

\subsection{Impossibility Results due to Weakly Fair Daemon}

In this section we prove there exists no universal $(1,r)-$ftss algorithm for minimal or priority \textbf{AU} under a weakly fair daemon for any $r$ value. 

The first impossibility result  uses the following property: if there exists a universal algorithm $\mathcal{A}$ that is $(1,r)-$ftss for minimal \textbf{AU} under a weakly fair daemon for a natural number $r$, then an arbitrary processor $p$ is not enabled for $\mathcal{A}$ if it has only one neighbor $p'$ and if $c_{p}=c_{p'}$ (proved in Lemma \ref{lem:impWFMin} formally stated below). Then, we show that $\mathcal{A}$ starves the network reduced to a two-correct-processor chain where all clock values are identical (see Proposition \ref{prop:impWFMin}).

\begin{lemma}\label{lem:impWFMin}
If there exists a universal algorithm $\mathcal{A}$ that is $(1,r)-$ftss for minimal \textbf{AU} under a weakly fair daemon for a natural number $r$, then an arbitrary processor $p$ is not enabled for $\mathcal{A}$ if it has only one neighbor $p'$ and if $c_{p}=c_{p'}$.
\end{lemma}

\begin{proof}
Let $r$ be a natural number. Let $\mathcal{A}$ be a universal $(1,r)-$ftss algorithm for the minimal \textbf{AU} under a weakly fair daemon.

Let $G$ be the network reduced to a chain of length $r+2$. Assume processors in $G$ labeled as follows: $p_{0},p_{1},\ldots,p_{r+2}$. Consider the following configurations of $G$ (see Figure \ref{fig:Figure1}):

\begin{itemize}
\item $\gamma_{1}$ defined by $\forall i\in \{0,\ldots,r+1\},c_{p_{i}}=i$ and $c_{p_{r+2}}=r+1$ and $p_{0}$ crashed.
\item $\gamma_{2}$ defined by $\forall i\in \{0,\ldots,r+1\},c_{p_{i}}=2r+2-i$ and $c_{p_{r+2}}=r+1$ and $p_{0}$ crashed.
\item $\gamma_{3}$ defined by $\forall i\in \{0,\ldots,r+2\},c_{p_{i}}=i$ and $p_{0}$ crashed.
\end{itemize}

\input{figure2}

By Lemma \ref{lem:blocage}, processors from $p_{1}$ to $p_{r}$ are not enabled in such configurations (and remain not enabled until one of the processors within $p_{0} \ldots p_{r+1}$ executes a rule).

Note that for the processor $p_{r+2}$, the configurations $\gamma_{1}$ and $\gamma_{2}$ are indistinguishable (otherwise the unison would not be minimal). We are going to prove the result by contradiction. Assume $p_{r+2}$ is enabled in $\gamma_{1}$ and $\gamma_{2}$. The safety property of $\mathcal{A}$ implies that the enabled rule for $p_{r+2}$ modifies its clock either to $r+2$ or to $r$. In the following we discuss these cases separately:

\begin{description}
\item[Case 1:] The enabled rule for $p_{r+2}$ modifies its clock into $r+2$. \\
Assume without loss of generality that $p_{r+2}$ is the only activated processor. Hence its clock takes the value $r+2$. The following cases are possible in the obtained configuration:

\begin{description}
\item[Case 1.1:] $p_{r+2}$ is not enabled.\\
If an execution started from $\gamma_{1}$, then no processor is enabled, which contradicts the liveness property of \textbf{AU}.
\item[Case 1.2 :] $p_{r+2}$ is enabled and the enabled rule modifies its clock into $r+1$.\\
Let $\epsilon$ be an execution starting from $\gamma_{1}$ where only $p_{r+2}$ is activated. Consequently, the clock of the processor $p_{r+2}$ takes infinitely the following sequence of values: $r+1,r+2$. In this execution, $p_{r+2}$ executes infinitely often while processors from $p_{0}$ to $p_{r}$ are never enabled. Note that $p_{r+1}$ is not enabled when $c_{p_{r+2}}=r+2$, hence this processor is never infinitely enabled. In conclusion, this execution is allowed by the weakly fair scheduler. Note that this execution starves $p_{r+1}$, which contradicts the liveness property of $\mathcal{A}$. 
\item[Case 1.3 :] $p_{r+2}$ is enabled and the enabled rule modifies its clock into $r$.\\
The execution of this rule leads to case 2.
\end{description}

\item[Case 2 :] The enabled rule for $p_{r+2}$ modifies its clock into $r$. \\
Assume without loss of generality that $p_{r+2}$ is the only activated processor and after its execution the new configuration satisfies one of the the following cases:

\begin{description}
\item[Case 2.1 :] $p_{r+2}$ is not enabled.\\
If an execution started from $\gamma_{2}$, then no processor is enabled, which contradicts the liveness property (the network is starved).
\item[Case 2.2 :] $p_{r+2}$ is enabled and the enabled rule modifies its clock into $r+1$.\\
Let $\epsilon$ be an execution starting from $\gamma_{2}$ that contains only actions of $p_{r+2}$ (its clock takes infinitely the following value sequence : $r+1,r$). In this execution, $p_{r+2}$ executes a rule infinitely often (by construction) and processors from $p_{0}$ to $p_{r}$ are never enabled. Note that $p_{r+1}$ is not enabled when $c_{p_{r+2}}=r$, so this processor is never infinitely enabled. In conclusion, this execution satisfies the weakly fair scheduling. 

Note that this execution starves $p_{r+1}$, which contradicts the liveness property of $\mathcal{A}$. 
\item[Case 2.3 :] $p_{r+2}$ is enabled and the enabled rule modifies its clock into $r+2$.\\
The execution of these rule leads to case 1.

\end{description}
\end{description}

Overall, the only two possible cases (cases 1.3 and 2.3) are the following: 

\begin{enumerate} 
\item $p_{r+2}$ is enabled for modifying its clock value into $r$ when $c_{p_{r+2}}=r+2$ and $c_{p_{r+1}}=r+1$.
\item  $p_{r+2}$ is enabled for modifying its clock value into $r+2$  when $c_{p_{r+2}}=r$ and $c_{p_{r+1}}=r+1$. 
\end{enumerate}

Let $\epsilon$ be an execution starting from $\gamma_{3}$ that contains only actions of $p_{r+2}$ (its clock takes infinitely the following sequence of values: $r+2,r$). In this execution, $p_{r+2}$ executes a rule infinitely often (by construction) and processors in $p_{0} \ldots p_{r}$ are never enabled. Note that $p_{r+1}$ is not enabled when $c_{p_{r+2}}=r+2$, so this processor is never infinitely enabled. In conclusion, this execution satisfies the weakly fair scheduling. 

This execution starves $p_{r+1}$, which contradicts the liveness property of $\mathcal{A}$ and proves the result.
\end{proof}

\begin{proposition}\label{prop:impWFMin}
For any natural number $r$, there exists no universal $(1,r)-$ftss algorithm for \emph{minimal} \textbf{AU} under a weakly fair daemon.
\end{proposition}

\begin{proof}
Let $r$ be a natural integer. Assume there exists a universal $(1,r)-$ftss algorithm $\mathcal{A}$ for the minimal \textbf{AU} under a weakly fair daemon. By Lemma \ref{lem:impWFMin}, an arbitrary processor $p$ is not enabled for $\mathcal{A}$ if it has only one neighbor $p'$ and if $c_{p}=c_{p'}$.

Let $G$ be a network reduced to a chain of 2 processors $p$ and $p'$. Let $\gamma$ be a configuration of $G$ where $c_{p}=c_{p'}$ with no crashed processor. Notice that no processor is enabled in $\gamma$ that contradicts the liveness property of $\mathcal{A}$ and proves the result.
\end{proof}

The second main result of this section is that there exists no universal $(1,r)-$ftss algorithm for priority \textbf{AU} under a weakly fair daemon for any natural number $r$ (see Proposition \ref{prop:impWFPri}).

We prove this result by contradiction. We construct an execution starting from the configuration $\gamma_{0}^{0}$ shown in Figure \ref{fig:Figure2} allowed by a weakly fair scheduler. We prove that this execution starves $p_{r+1}$ that contradicts the liveness property of the algorithm. 

\begin{proposition}\label{prop:impWFPri}
For any natural number $r$, there exists no universal $(1,r)-$ftss algorithm for \emph{priority} \textbf{AU} under a weakly fair daemon.
\end{proposition}

\begin{proof}
Let $r$ be a natural number. Assume that there exists a universal $(1,r)-$ftss algorithm $\mathcal{A}$ for priority \textbf{AU} under a weakly fair daemon. Let $G$ be the network reduced to a chain of length $r+2$. Assume that processors in $G$ are labeled as follows: $p_{0},p_{1},\ldots,p_{r+2}$. Let $\gamma^{0}_{0}$ be a configuration such that $p_{0}$ is crashed and $\forall i\in\{0,\ldots,r+2\},c_{p_{i}}=i$ (See Figure \ref{fig:Figure2}). Note that all the other variables may have any value.

		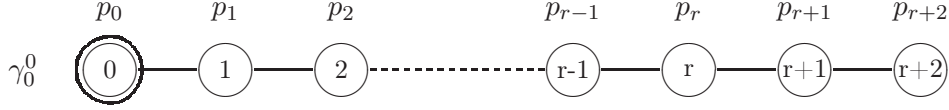
\begin{figure}
		\noindent \begin{centering}
			\ifx\JPicScale\undefined\def\JPicScale{0.77}\fi
			\unitlength \JPicScale mm
			\begin{picture}(165,15)(0,0)
			\linethickness{0.3mm}
			\put(40,5){\circle{10}}
			\linethickness{0.3mm}
			\put(60,5){\circle{10}}
			\linethickness{0.3mm}
			\put(160,5){\circle{10}}
			\linethickness{0.3mm}
			\put(140,5){\circle{10}}
			\linethickness{0.3mm}
			\put(120,5){\circle{10}}
			\linethickness{0.3mm}
			\put(100,5){\circle{10}}
			\linethickness{0.3mm}
			\put(25,5){\line(1,0){10}}
			\linethickness{0.3mm}
			\put(45,5){\line(1,0){10}}
			\linethickness{0.3mm}
			\put(105,5){\line(1,0){10}}
			\linethickness{0.3mm}
			\put(125,5){\line(1,0){10}}
			\linethickness{0.3mm}
			\put(145,5){\line(1,0){10}}
			\linethickness{0.3mm}
			\put(20,5){\circle{10}}
			\linethickness{0.3mm}
			\multiput(65,5)(1.94,0){16}{\line(1,0){0.97}}
			\linethickness{0.3mm}
			\put(25.6,4.75){\line(0,1){0.5}}
			\multiput(25.56,4.25)(0.04,0.5){1}{\line(0,1){0.5}}
			\multiput(25.48,3.75)(0.09,0.49){1}{\line(0,1){0.49}}
			\multiput(25.35,3.27)(0.13,0.49){1}{\line(0,1){0.49}}
			\multiput(25.18,2.8)(0.17,0.47){1}{\line(0,1){0.47}}
			\multiput(24.97,2.34)(0.1,0.23){2}{\line(0,1){0.23}}
			\multiput(24.72,1.91)(0.12,0.22){2}{\line(0,1){0.22}}
			\multiput(24.44,1.49)(0.14,0.21){2}{\line(0,1){0.21}}
			\multiput(24.12,1.11)(0.11,0.13){3}{\line(0,1){0.13}}
			\multiput(23.77,0.75)(0.12,0.12){3}{\line(0,1){0.12}}
			\multiput(23.39,0.43)(0.13,0.11){3}{\line(1,0){0.13}}
			\multiput(22.98,0.14)(0.2,0.14){2}{\line(1,0){0.2}}
			\multiput(22.55,-0.11)(0.21,0.13){2}{\line(1,0){0.21}}
			\multiput(22.1,-0.32)(0.22,0.11){2}{\line(1,0){0.22}}
			\multiput(21.63,-0.49)(0.47,0.17){1}{\line(1,0){0.47}}
			\multiput(21.16,-0.62)(0.48,0.13){1}{\line(1,0){0.48}}
			\multiput(20.67,-0.71)(0.49,0.09){1}{\line(1,0){0.49}}
			\multiput(20.17,-0.75)(0.49,0.04){1}{\line(1,0){0.49}}
			\put(19.68,-0.75){\line(1,0){0.5}}
			\multiput(19.18,-0.71)(0.49,-0.04){1}{\line(1,0){0.49}}
			\multiput(18.69,-0.62)(0.49,-0.09){1}{\line(1,0){0.49}}
			\multiput(18.22,-0.49)(0.48,-0.13){1}{\line(1,0){0.48}}
			\multiput(17.75,-0.32)(0.47,-0.17){1}{\line(1,0){0.47}}
			\multiput(17.3,-0.11)(0.22,-0.11){2}{\line(1,0){0.22}}
			\multiput(16.87,0.14)(0.21,-0.13){2}{\line(1,0){0.21}}
			\multiput(16.46,0.43)(0.2,-0.14){2}{\line(1,0){0.2}}
			\multiput(16.08,0.75)(0.13,-0.11){3}{\line(1,0){0.13}}
			\multiput(15.73,1.11)(0.12,-0.12){3}{\line(0,-1){0.12}}
			\multiput(15.41,1.49)(0.11,-0.13){3}{\line(0,-1){0.13}}
			\multiput(15.13,1.91)(0.14,-0.21){2}{\line(0,-1){0.21}}
			\multiput(14.88,2.34)(0.12,-0.22){2}{\line(0,-1){0.22}}
			\multiput(14.67,2.8)(0.1,-0.23){2}{\line(0,-1){0.23}}
			\multiput(14.5,3.27)(0.17,-0.47){1}{\line(0,-1){0.47}}
			\multiput(14.37,3.75)(0.13,-0.49){1}{\line(0,-1){0.49}}
			\multiput(14.29,4.25)(0.09,-0.49){1}{\line(0,-1){0.49}}
			\multiput(14.25,4.75)(0.04,-0.5){1}{\line(0,-1){0.5}}
			\put(14.25,4.75){\line(0,1){0.5}}
			\multiput(14.25,5.25)(0.04,0.5){1}{\line(0,1){0.5}}
			\multiput(14.29,5.75)(0.09,0.49){1}{\line(0,1){0.49}}
			\multiput(14.37,6.25)(0.13,0.49){1}{\line(0,1){0.49}}
			\multiput(14.5,6.73)(0.17,0.47){1}{\line(0,1){0.47}}
			\multiput(14.67,7.2)(0.1,0.23){2}{\line(0,1){0.23}}
			\multiput(14.88,7.66)(0.12,0.22){2}{\line(0,1){0.22}}
			\multiput(15.13,8.09)(0.14,0.21){2}{\line(0,1){0.21}}
			\multiput(15.41,8.51)(0.11,0.13){3}{\line(0,1){0.13}}
			\multiput(15.73,8.89)(0.12,0.12){3}{\line(0,1){0.12}}
			\multiput(16.08,9.25)(0.13,0.11){3}{\line(1,0){0.13}}
			\multiput(16.46,9.57)(0.2,0.14){2}{\line(1,0){0.2}}
			\multiput(16.87,9.86)(0.21,0.13){2}{\line(1,0){0.21}}
			\multiput(17.3,10.11)(0.22,0.11){2}{\line(1,0){0.22}}
			\multiput(17.75,10.32)(0.47,0.17){1}{\line(1,0){0.47}}
			\multiput(18.22,10.49)(0.48,0.13){1}{\line(1,0){0.48}}
			\multiput(18.69,10.62)(0.49,0.09){1}{\line(1,0){0.49}}
			\multiput(19.18,10.71)(0.49,0.04){1}{\line(1,0){0.49}}
			\put(19.68,10.75){\line(1,0){0.5}}
			\multiput(20.17,10.75)(0.49,-0.04){1}{\line(1,0){0.49}}
			\multiput(20.67,10.71)(0.49,-0.09){1}{\line(1,0){0.49}}
			\multiput(21.16,10.62)(0.48,-0.13){1}{\line(1,0){0.48}}
			\multiput(21.63,10.49)(0.47,-0.17){1}{\line(1,0){0.47}}
			\multiput(22.1,10.32)(0.22,-0.11){2}{\line(1,0){0.22}}
			\multiput(22.55,10.11)(0.21,-0.13){2}{\line(1,0){0.21}}
			\multiput(22.98,9.86)(0.2,-0.14){2}{\line(1,0){0.2}}
			\multiput(23.39,9.57)(0.13,-0.11){3}{\line(1,0){0.13}}
			\multiput(23.77,9.25)(0.12,-0.12){3}{\line(0,-1){0.12}}
			\multiput(24.12,8.89)(0.11,-0.13){3}{\line(0,-1){0.13}}
			\multiput(24.44,8.51)(0.14,-0.21){2}{\line(0,-1){0.21}}
			\multiput(24.72,8.09)(0.12,-0.22){2}{\line(0,-1){0.22}}
			\multiput(24.97,7.66)(0.1,-0.23){2}{\line(0,-1){0.23}}
			\multiput(25.18,7.2)(0.17,-0.47){1}{\line(0,-1){0.47}}
			\multiput(25.35,6.73)(0.13,-0.49){1}{\line(0,-1){0.49}}
			\multiput(25.48,6.25)(0.09,-0.49){1}{\line(0,-1){0.49}}
			\multiput(25.56,5.75)(0.04,-0.5){1}{\line(0,-1){0.5}}
			\put(5,5){\makebox(0,0)[cc]{$\gamma_{0}^{0}$}}
			\put(20,5){\makebox(0,0)[cc]{\small{0}}}
			\put(40,5){\makebox(0,0)[cc]{\small{1}}}
			\put(60,5){\makebox(0,0)[cc]{\small{2}}}
			\put(140,5){\makebox(0,0)[cc]{\small{r+1}}}
			\put(120,5){\makebox(0,0)[cc]{\small{r}}}
			\put(100,5){\makebox(0,0)[cc]{\small{r-1}}}
			\put(20,15){\makebox(0,0)[cc]{$p_{0}$}}
			\put(40,15){\makebox(0,0)[cc]{$p_{1}$}}
			\put(60,15){\makebox(0,0)[cc]{$p_{2}$}}
			\put(100,15){\makebox(0,0)[cc]{$p_{r-1}$}}
			\put(120,15){\makebox(0,0)[cc]{$p_{r}$}}
			\put(140,15){\makebox(0,0)[cc]{$p_{r+1}$}}
			\put(160,15){\makebox(0,0)[cc]{$p_{r+2}$}}
			\put(160,5){\makebox(0,0)[cc]{\small{r+2}}}
			\end{picture}
			\par\end{centering}\caption{\label{fig:Figure2}Initial configuration used in the proof of Proposition \ref{prop:impWFPri} 
																	(the numbers represent clock values and the double circles represent  crashed processor).}
		\end{figure}

We construct a fragment of execution $\epsilon_{0}'=\gamma^{0}_{0}\gamma^{0}_{1}\gamma^{0}_{2}\ldots\gamma^{0}_{r+1}$ starting from $\gamma^{0}_{0}$ such that $\forall i\in\{0,1,\ldots,r\}$, the step $\gamma^{0}_{i}\rightarrow\gamma^{0}_{i+1}$ contains only an action of $p_{i+1}$ if $p_{i+1}$ is enabled. By Lemma \ref{lem:blocage}, this fragment does not modify the clock value of any processor in $\{p_{0} \ldots p_{r+1}\}$.

We also construct a fragment of execution, $\epsilon_{0}''$, starting from $\gamma^{0}_{r+1}$ using the following cases:

\begin{description}
\item[Case 1:] $p_{r+2}$ is not enabled in $\gamma^{0}_{r+1}$.\\
Let $\epsilon_{0}''$ be $\epsilon$ (empty word).
\item[Case 2:] $p_{r+2}$ is enabled in $\gamma^{0}_{r+1}$.\\
We distinguish now the following sub-cases:
\begin{description}
\item[Case 2.1:] There exists a rule of $p_{r+2}$ enabled in $\gamma^{0}_{r+1}$ that does not modify the clock value of $p_{r+2}$.\\
Let $\epsilon_{0}''$ be $\gamma^{0}_{r+1}\gamma^{0}_{r+2}$ where step $\gamma^{0}_{r+1}\rightarrow\gamma^{0}_{r+2}$ contains only the execution of this rule by $p_{r+2}$.
\item[Case 2.2:] Any enabled rule of $p_{r+2}$ in $\gamma^{0}_{r+1}$ modifies its clock value.\\
Note that the safety property of $\mathcal{A}$ implies that the clock of $p_{r+2}$ takes the value $r$ or $r+1$. Let us study the following cases.
\begin{description}
\item[Case 2.2.1:]  There exists a rule of $p_{r+2}$ enabled in $\gamma^{0}_{r+1}$ that modifies its clock value into $r+1$.\\
Since $\mathcal{A}$ is a priority unison, there exists by definition a fragment of execution $\epsilon_{0}''=\gamma^{0}_{r+1}\gamma^{0}_{r+2}\ldots\gamma^{0}_{r+k}$  that contains only actions of $p_{r+2}$ such that (i) $p_{r+2}$ executes one of the rules that modifies its clock value into $r+1$ in the step $\gamma^{0}_{r+1}\rightarrow\gamma^{0}_{r+2}$ (ii) in the steps from $\gamma^{0}_{r+2}$ to $\gamma^{0}_{r+k-1}$ the clock value of $p_{r+2}$ is not modified while (iii) in the step	$\gamma^{0}_{r+k-1}\rightarrow\gamma^{0}_{r+k}$ the clock value of $p_{r+2}$ is incremented.
\item[Case 2.2.2:]  Any enabled rule of $p_{r+2}$ in $\gamma^{0}_{r+1}$ modifies its clock value into $r$.\\
Since $\mathcal{A}$ is a priority unison, there exists by definition a fragment of execution $\epsilon_{a}=\gamma^{0}_{r+1}\gamma^{0}_{r+2}\ldots\gamma^{0}_{r+k}$ that contains only actions of $p_{r+2}$ such that (i) $p_{r+2}$ executes one of the rules that modifies its clock value into $r$ in the step $\gamma^{0}_{r+1}\rightarrow\gamma^{0}_{r+2}$ (ii) in the steps from $\gamma^{0}_{r+2}$ to $\gamma^{0}_{r+k-1}$ the clock value of $p_{r+2}$ is not modified and (iii) in the step	$\gamma^{0}_{r+k-1}\rightarrow\gamma^{0}_{r+k}$ the clock of $p_{r+2}$ takes the value $r+1$. 	
	
Since $\mathcal{A}$ is a priority unison, there exists by definition a fragment of execution $\epsilon_{b}=\gamma^{0}_{r+k}\gamma^{0}_{r+k+1}\ldots\gamma^{0}_{r+j}$ that contains only actions of $p_{r+2}$ such that (i) in the steps from $\gamma^{0}_{r+k}$ to $\gamma^{0}_{r+j-1}$ the clock value of $p_{r+2}$ is not modified and (ii) in the step	$\gamma^{0}_{r+j-1}\rightarrow\gamma^{0}_{r+j}$ the clock value of $p_{r+2}$ is incremented.						

Let $\epsilon_{0}''$ be $\epsilon_{a}\epsilon_{b}$.
\end{description}
\end{description}
\end{description}

In all cases, we construct a fragment of execution $\epsilon_{0}=\epsilon_{0}'\epsilon_{0}''$ such that its last configuration (let us denote it by $\gamma_{0}^{1}$) satisfies: the value of any clock is identical to the one in $\gamma^{0}_{0}$ (the others variables may have changed). Then, we can reiterate the reasoning and obtain a fragment of execution $\epsilon_{1},\epsilon_{2}\ldots$ (respectively starting from $\gamma_{0}^{1},\gamma_{0}^{2},\ldots$) that satisfies the same property.

We finally obtain an execution $\epsilon=\epsilon_{0}\epsilon_{1}\ldots$ that satisfies:

\begin{itemize}
\item No processor is infinitely enabled without executing a rule (since all enabled processors in $\gamma^{i}_{0}$ execute a rule or are neutralized during $\epsilon_{i}$). Consequently $\epsilon$ is an execution that satisfies the  weakly fair scheduling.
\item The clock of processor $p_{r+1}$ never changes (whereas $d(p_{0},p_{r+1})=r+1$).
\end{itemize}

This execution contradicts the liveness property of $\mathcal{A}$ that is a $(1,r)-$ftss algorithm for priority \textbf{AU} under a weakly fair daemon by hypothesis. 
\end{proof}

\subsection{Impossibility Results due to Strongly Fair Daemon}

In this section we prove that there exists no universal $(1,r)-$ftss algorithm for minimal or priority \textbf{AU} under a strongly fair daemon if the degree of the network is at least 3. 

In order to prove the first impossibility result, we use the following property: if a processor $p$ has only one neighbor $q$ such that $c_{q}=r+1$	and if $|c_{p}-c_{q}|\leq 1$, then $p$ is enabled in any universal $(1,r)-$ftss algorithm for minimal \textbf{AU} (see Lemma \ref{lem:impSFMin}). Then we construct a strongly fair infinite execution that starves a processor such that the closest crashed processor is at more than $r$ hops away. This execution contradicts the liveness property of the \textbf{AU} problem (see Proposition \ref{prop:impSFMin}).

\begin{lemma}\label{lem:impSFMin}
Let $\mathcal{A}$ be a universal $(1,r)-$ftss algorithm for minimal \textbf{AU}. If a processor $p$ has only one neighbor $q$ such that $c_{q}=r+1$ and if $|c_{p}-c_{q}|\leq 1$, then $p$ is enabled in $\mathcal{A}$. 
\end{lemma}

\begin{proof}
Assume that there exists a universal algorithm $\mathcal{A}$ that is $(1,r)-$ftss for minimal \textbf{AU}. Let $G$ be a network that executes $\mathcal{A}$ and that contains at least one processor $p$ that has only one neighbor $q$. Assume that $c_{q}=r+1$ and $|c_{p}-c_{q}|\leq 1$. Then, we have:

\begin{enumerate}
\item If $c_{p}=r$, then $p$ is enabled for at least one rule of $\mathcal{A}$. Otherwise, all processors are starved in the network reduced to the chain $p_{0},\ldots,p_{r},q,p$ in the configuration $\gamma_{1}$ defined by $\forall i\in\{0,\ldots,r\},c_{p_{i}}=2r+2-i$, $c_{q}=r+1$, $c_{p}=r$ where $p_{0}$ is crashed (see Figure \ref{fig:Figure3}) since no correct processor is enabled (by Lemma \ref{lem:blocage}).
\item If $c_{p}=r+1$, then $p$ is enabled for at least one rule of $\mathcal{A}$. Otherwise, all processors are starved in the network reduced to the	chain $q,p$ in the configuration $\gamma_{2}$ defined by $c_{q}=c_{p}=r+1$ and where no processor is crashed (see Figure \ref{fig:Figure3}). Indeed, the symmetry of the configuration implies that $q$ is enabled if and only if $p$ is enabled.
\item If $c_{p}=r+2$, then $p$ is enabled for at least one rule of $\mathcal{A}$. Otherwise, all processors are starved in the network reduced to the	chain $p_{0},\ldots,p_{r},q,p$ in the configuration $\gamma_{3}$ defined by $\forall i\in\{0,\ldots,r\},c_{p_{i}}i$, $c_{q}=r+1$, $c_{p}=r+2$ and $p_{0}$ crashed (see Figure \ref{fig:Figure3}) since no correct processor is enabled (by Lemma \ref{lem:blocage}).
\end{enumerate}

\input{figure4}

\end{proof}

\begin{proposition}\label{prop:impSFMin}
For any natural number $r$, there exists no universal $(1,r)-$ftss algorithm for \emph{minimal} \textbf{AU} under a strongly fair daemon if the system has a maximal degree of at least 3.
\end{proposition}

\begin{proof}
Let $r$ be a natural number. Assume that there exists a universal $(1,r)-$ftss algorithm $\mathcal{A}$ for the minimal \textbf{AU} under a strongly fair daemon in a network with a degree of at least 3. Let $G$ be the network defined by: $V=\{p_{0},\ldots,p_{r+1},q,q'\}$ and $E=\{\{p_{i},p_{i+1}\},i\in\{0,\ldots,r\}\}\cup\{\{p_{r+1},q\},\{p_{r+1},q'\}\}$.

As $\mathcal{A}$ is deterministic and the system anonymous, $q$ and $q'$ must behave identically if they have the same clock value (in this case, their local configurations are identical). If $c_{p_{r+1}}=r+1$ and $|c_{p_{r+1}}-c_{q}|\leq 1$, there exists three local configurations for $q$: (1) $c_{q}=r$, (2) $c_{q}=r+1$ or (3) $c_{q}=r+2$ (the same property holds for $q'$).

By Lemma \ref{lem:impSFMin}, processor $q$ (respectively $q'$) is enabled in any configuration where $c_{p_{r+1}}=r+1$ and $|c_{p_{r+1}}-c_{q}|\leq 1$ (respectively $|c_{p_{r+1}}-c_{q'}|\leq 1$). Moreover, in this case, the enabled rule for $q$ (respectively $q'$) modifies its clock into a value in $\{r,r+1,r+2\}-\{c_{q}\}$ (respectively $\{r,r+1,r+2\}-\{c_{q'}\}$) by the safety property of $\mathcal{A}$.

For each of the three possible local configurations for $q$ or $q'$ (studied in the proof of Lemma \ref{lem:impSFMin}), $\mathcal{A}$ can only allow $2$ moves. Hence, there exists $8$ possible moves for $\mathcal{A}$. Let us denote each of these possibilities by a triplet $(a,b,c)$ where $a$, $b$ and $c$ are the clock value of $q$ after the allowed move when $c_{q}=r$, $c_{q}=r+1$, and $c_{q}=r+2$ respectively. Note that, due to the determinism of $\mathcal{A}$, moves allowed for $q'$ and $q$ are identical. There exists the following cases:

\input{figure5}

\begin{description}
\item[Case 1:] $(r+1,r,r)$\\
Let $\gamma_{1}$ be the configuration of $G$ defined by: $\forall i\in\{0,\ldots,r+1\},c_{p_{i}}=2r+2-i$, $c_{q}=r+1$ and $c_{q'}=r$ and $p_{0}$ crashed (see Figure \ref{fig:Figure4}). Note that only $q$ and $q'$ are enabled (by Lemma \ref{lem:blocage}). Assume $q$ executes. Hence, its clock takes the value $r$. By Lemma \ref{lem:blocage}, only $q$ and $q'$ are enabled. Assume now that $q'$ executes. Its clock takes the value $r+1$. This configuration is identical to $\gamma_{1}$ (since processors are anonymous), we can repeat the above reasoning in order to obtain an infinite execution where processors $p_{1}, \ldots, p_{r+1}$ are never enabled (see Figure \ref{fig:Figure5} for an illustration when $r=1$).

\input{figure6}

\item[Case 2:] $(r+1,r+2,r)$\\
Let $\gamma_{2}$ be the configuration of $G$ defined by: $\forall i\in\{0,\ldots,r+1\},c_{p_{i}}i$, $c_{q}=r$ and $c_{q'}=r+2$ and $p_{0}$ crashed (see Figure \ref{fig:Figure4}). Note that only $q$ and $q'$ are enabled (by Lemma \ref{lem:blocage}). Assume $q$ executes. Its clock takes the value $r+1$. By Lemma \ref{lem:blocage}, only $q$ and $q'$ are enabled. Assume $q$ executes its rule again.  Its clock takes the value $r+2$. By Lemma \ref{lem:blocage}, only $q$ and $q'$ are enabled. Assume now that $q'$ executes its rule. Its clock takes the value $r$. This configuration is identical to $\gamma_{2}$ (since processors are anonymous). We can repeat the reasoning in order to obtain an infinite execution where processors in $p_{1},\ldots,p_{r+1}$ are never enabled.

\item[Case 3:] $(r+1,r,r+1)$\\
Similar to the reasoning of case 1.

\item[Case 4:] $(r+1,r+2,r+1)$\\
Let $\gamma_{3}$ be the configuration of $G$ defined by: $\forall i\in\{0,\ldots,r+1\},c_{p_{i}}=i$, $c_{q}=r+2$ and $c_{q'}=r+1$ and where $p_{0}$ is crashed (see Figure \ref{fig:Figure4}). Note that only $q$ and $q'$ are enabled (by Lemma \ref{lem:blocage}). Assume $q'$ executes its rule. Its clock takes the value $r+2$. By Lemma \ref{lem:blocage}, only $q$ and $q'$ are enabled. Assume now that $q$ executes its rule. Its clock takes the value $r+1$. This configuration is identical to $\gamma_{3}$ (since processors are anonymous). We can repeat the reasoning in order to obtain an infinite execution where processors in $p_{1}, \ldots, p_{r+1}$ are never enabled.

\item[Case 5:] $(r+2,r,r)$\\
Let $\gamma_{2}$ be the configuration of $G$ as defined in the case 2 above. Note that only $q$ and $q'$ are enabled (by Lemma \ref{lem:blocage}). Assume $q$ executes its rule. Its clock takes the value $r+2$. By Lemma \ref{lem:blocage}, only $q$ and $q'$ are enabled. Assume now that $q'$ executes its rule. Its clock takes the value $r$. This configuration is identical to $\gamma_{2}$ (since processors are anonymous). We can repeat the reasoning in order to obtain an infinite execution where processors $p_{1}, \ldots, p_{r+1}$ are never enabled.

\item[Case 6:] $(r+2,r+2,r)$\\
The reasoning is similar to the case 5.

\item[Case 7:] $(r+2,r,r+1)$\\
Let $\gamma_{2}$ be the configuration of $G$ as defined in the case 2 above. Note that only $q$ and $q'$ are enabled (by Lemma \ref{lem:blocage}). Assume $q$ executes its rule. Its clock takes the value  $r+2$. By Lemma \ref{lem:blocage}, only $q$ and $q'$ are enabled. Assume $q'$ executes its rule. Its clock takes the value $r+1$. By Lemma \ref{lem:blocage}, only $q$ and $q'$ are enabled. Assume $q'$ executes again its rule. Its clock takes the value $r$. This configuration is identical to $\gamma_{2}$ (since processors are anonymous). We can repeat the above scenario in order to obtain an infinite execution where processors $p_{1},\ldots,p_{r+1}$ are never enabled.

\item[Case 8:] $(r+2,r+2,r+1)$\\
The proof is similar to the case 4.

\end{description}

Overall, we can construct an infinite execution where processor $p_{0}$ is crashed, processors from $p_{1}$ to $p_{r+1}$ are never enabled and processors $q$ and $q'$ execute a rule infinitely often. This execution satisfies the strongly fair scheduling. Notice that in this execution  $p_{r+1}$ is never enabled, hence it is starved. This contradicts the liveness property of $\mathcal{A}$ and proves the result.
\end{proof}

The second main result of this section is that there exists no universal $(1,r)-$ftss algorithm for priority \textbf{AU} under a strongly fair daemon for any natural number $r$ if the degree of the graph modeling the network is at least 3. (see Proposition \ref{prop:impSFPri}).

We prove this result by contradiction. We construct an execution starting from the configuration $\gamma_{0}^{0}$ of Figure \ref{fig:Figure6} satisfying the  strongly fair scheduling that starves $p_{r+1}$, which contradicts the liveness of the algorithm. 

\begin{proposition}\label{prop:impSFPri}
For any natural number $r$, there exists no universal $(1,r)-$ftss algorithm for \emph{priority} \textbf{AU} under a strongly fair daemon if the system has a maximal degree of at least 3.
\end{proposition}

\begin{proof}
Let $r$ be a natural number. Assume that there exists a universal $(1,r)-$ftss algorithm $\mathcal{A}$  for priority \textbf{AU} under a strongly fair daemon even if the graph modeling the network has a degree of at least 3. Let $G$ be the network defined by: $V=\{p_{0},\ldots,p_{r+1},q,q'\}$ and $E=\{\{p_{i},p_{i+1}\},i\in\{0,\ldots,r\}\}\cup\{\{p_{r+1},q\},\{p_{r+1},q'\}\}$. Note that $G$ has a degree equal to $3$.

		\begin{figure}
			\noindent \begin{centering}
			\ifx\JPicScale\undefined\def\JPicScale{0.77}\fi
			\unitlength \JPicScale mm
			\begin{picture}(170,30)(0,0)
			\linethickness{0.3mm}
			\put(20,15){\circle{10}}
			\linethickness{0.3mm}
			\put(40,15){\circle{10}}
			\linethickness{0.3mm}
			\put(60,15){\circle{10}}
			\linethickness{0.3mm}
			\put(160,25){\circle{10}}
			\linethickness{0.3mm}
			\put(140,15){\circle{10}}
			\linethickness{0.3mm}
			\put(120,15){\circle{10}}
			\linethickness{0.3mm}
			\put(100,15){\circle{10}}
			\linethickness{0.3mm}
			\put(25,15){\line(1,0){10}}
			\linethickness{0.3mm}
			\put(45,15){\line(1,0){10}}
			\linethickness{0.3mm}
			\put(105,15){\line(1,0){10}}
			\linethickness{0.3mm}
			\put(125,15){\line(1,0){10}}
			\linethickness{0.3mm}
			\multiput(145,15)(0.12,-0.12){83}{\line(1,0){0.12}}
			\linethickness{0.3mm}
			\multiput(65,15)(1.94,0){16}{\line(1,0){0.97}}
			\put(20,25){\makebox(0,0)[cc]{$p_{0}$}}
			\put(40,25){\makebox(0,0)[cc]{$p_{1}$}}
			\put(60,25){\makebox(0,0)[cc]{$p_{2}$}}
			\put(100,25){\makebox(0,0)[cc]{$p_{r-1}$}}
			\put(120,25){\makebox(0,0)[cc]{$p_{r}$}}
			\linethickness{0.3mm}
			\put(25.75,14.75){\line(0,1){0.5}}
			\multiput(25.71,14.25)(0.04,0.5){1}{\line(0,1){0.5}}
			\multiput(25.63,13.75)(0.09,0.49){1}{\line(0,1){0.49}}
			\multiput(25.5,13.27)(0.13,0.49){1}{\line(0,1){0.49}}
			\multiput(25.33,12.8)(0.17,0.47){1}{\line(0,1){0.47}}
			\multiput(25.12,12.34)(0.1,0.23){2}{\line(0,1){0.23}}
			\multiput(24.87,11.91)(0.12,0.22){2}{\line(0,1){0.22}}
			\multiput(24.59,11.49)(0.14,0.21){2}{\line(0,1){0.21}}
			\multiput(24.27,11.11)(0.11,0.13){3}{\line(0,1){0.13}}
			\multiput(23.92,10.75)(0.12,0.12){3}{\line(0,1){0.12}}
			\multiput(23.54,10.43)(0.13,0.11){3}{\line(1,0){0.13}}
			\multiput(23.13,10.14)(0.2,0.14){2}{\line(1,0){0.2}}
			\multiput(22.7,9.89)(0.21,0.13){2}{\line(1,0){0.21}}
			\multiput(22.25,9.68)(0.22,0.11){2}{\line(1,0){0.22}}
			\multiput(21.78,9.51)(0.47,0.17){1}{\line(1,0){0.47}}
			\multiput(21.31,9.38)(0.48,0.13){1}{\line(1,0){0.48}}
			\multiput(20.82,9.29)(0.49,0.09){1}{\line(1,0){0.49}}
			\multiput(20.32,9.25)(0.49,0.04){1}{\line(1,0){0.49}}
			\put(19.83,9.25){\line(1,0){0.5}}
			\multiput(19.33,9.29)(0.49,-0.04){1}{\line(1,0){0.49}}
			\multiput(18.84,9.38)(0.49,-0.09){1}{\line(1,0){0.49}}
			\multiput(18.37,9.51)(0.48,-0.13){1}{\line(1,0){0.48}}
			\multiput(17.9,9.68)(0.47,-0.17){1}{\line(1,0){0.47}}
			\multiput(17.45,9.89)(0.22,-0.11){2}{\line(1,0){0.22}}
			\multiput(17.02,10.14)(0.21,-0.13){2}{\line(1,0){0.21}}
			\multiput(16.61,10.43)(0.2,-0.14){2}{\line(1,0){0.2}}
			\multiput(16.23,10.75)(0.13,-0.11){3}{\line(1,0){0.13}}
			\multiput(15.88,11.11)(0.12,-0.12){3}{\line(0,-1){0.12}}
			\multiput(15.56,11.49)(0.11,-0.13){3}{\line(0,-1){0.13}}
			\multiput(15.28,11.91)(0.14,-0.21){2}{\line(0,-1){0.21}}
			\multiput(15.03,12.34)(0.12,-0.22){2}{\line(0,-1){0.22}}
			\multiput(14.82,12.8)(0.1,-0.23){2}{\line(0,-1){0.23}}
			\multiput(14.65,13.27)(0.17,-0.47){1}{\line(0,-1){0.47}}
			\multiput(14.52,13.75)(0.13,-0.49){1}{\line(0,-1){0.49}}
			\multiput(14.44,14.25)(0.09,-0.49){1}{\line(0,-1){0.49}}
			\multiput(14.4,14.75)(0.04,-0.5){1}{\line(0,-1){0.5}}
			\put(14.4,14.75){\line(0,1){0.5}}
			\multiput(14.4,15.25)(0.04,0.5){1}{\line(0,1){0.5}}
			\multiput(14.44,15.75)(0.09,0.49){1}{\line(0,1){0.49}}
			\multiput(14.52,16.25)(0.13,0.49){1}{\line(0,1){0.49}}
			\multiput(14.65,16.73)(0.17,0.47){1}{\line(0,1){0.47}}
			\multiput(14.82,17.2)(0.1,0.23){2}{\line(0,1){0.23}}
			\multiput(15.03,17.66)(0.12,0.22){2}{\line(0,1){0.22}}
			\multiput(15.28,18.09)(0.14,0.21){2}{\line(0,1){0.21}}
			\multiput(15.56,18.51)(0.11,0.13){3}{\line(0,1){0.13}}
			\multiput(15.88,18.89)(0.12,0.12){3}{\line(0,1){0.12}}
			\multiput(16.23,19.25)(0.13,0.11){3}{\line(1,0){0.13}}
			\multiput(16.61,19.57)(0.2,0.14){2}{\line(1,0){0.2}}
			\multiput(17.02,19.86)(0.21,0.13){2}{\line(1,0){0.21}}
			\multiput(17.45,20.11)(0.22,0.11){2}{\line(1,0){0.22}}
			\multiput(17.9,20.32)(0.47,0.17){1}{\line(1,0){0.47}}
			\multiput(18.37,20.49)(0.48,0.13){1}{\line(1,0){0.48}}
			\multiput(18.84,20.62)(0.49,0.09){1}{\line(1,0){0.49}}
			\multiput(19.33,20.71)(0.49,0.04){1}{\line(1,0){0.49}}
			\put(19.83,20.75){\line(1,0){0.5}}
			\multiput(20.32,20.75)(0.49,-0.04){1}{\line(1,0){0.49}}
			\multiput(20.82,20.71)(0.49,-0.09){1}{\line(1,0){0.49}}
			\multiput(21.31,20.62)(0.48,-0.13){1}{\line(1,0){0.48}}
			\multiput(21.78,20.49)(0.47,-0.17){1}{\line(1,0){0.47}}
			\multiput(22.25,20.32)(0.22,-0.11){2}{\line(1,0){0.22}}
			\multiput(22.7,20.11)(0.21,-0.13){2}{\line(1,0){0.21}}
			\multiput(23.13,19.86)(0.2,-0.14){2}{\line(1,0){0.2}}
			\multiput(23.54,19.57)(0.13,-0.11){3}{\line(1,0){0.13}}
			\multiput(23.92,19.25)(0.12,-0.12){3}{\line(0,-1){0.12}}
			\multiput(24.27,18.89)(0.11,-0.13){3}{\line(0,-1){0.13}}
			\multiput(24.59,18.51)(0.14,-0.21){2}{\line(0,-1){0.21}}
			\multiput(24.87,18.09)(0.12,-0.22){2}{\line(0,-1){0.22}}
			\multiput(25.12,17.66)(0.1,-0.23){2}{\line(0,-1){0.23}}
			\multiput(25.33,17.2)(0.17,-0.47){1}{\line(0,-1){0.47}}
			\multiput(25.5,16.73)(0.13,-0.49){1}{\line(0,-1){0.49}}
			\multiput(25.63,16.25)(0.09,-0.49){1}{\line(0,-1){0.49}}
			\multiput(25.71,15.75)(0.04,-0.5){1}{\line(0,-1){0.5}}
			\put(140,25){\makebox(0,0)[cc]{$p_{r+1}$}}
			\linethickness{0.3mm}
			\put(160,5){\circle{10}}
			\linethickness{0.3mm}
			\multiput(145,15)(0.12,0.12){83}{\line(1,0){0.12}}
			\put(170,25){\makebox(0,0)[cc]{$q$}}
			\put(170,5){\makebox(0,0)[cc]{$q'$}}
			\put(20,15){\makebox(0,0)[cc]{\small{0}}}
			\put(40,15){\makebox(0,0)[cc]{\small{1}}}
			\put(60,15){\makebox(0,0)[cc]{\small{2}}}
			\put(100,15){\makebox(0,0)[cc]{\small{r-1}}}
			\put(120,15){\makebox(0,0)[cc]{\small{r}}}
			\put(140,15){\makebox(0,0)[cc]{\small{r+1}}}
			\put(160,25){\makebox(0,0)[cc]{\small{r+2}}}
			\put(5,15){\makebox(0,0)[cc]{$\gamma_{0}^{0}$}}
			\put(160,5){\makebox(0,0)[cc]{\small{r+2}}}
			\end{picture}
			\par\end{centering}\caption{\label{fig:Figure6}The initial configuration for the proof of Proposition \ref{prop:impSFPri} (the numbers represent clock values and the double circles represent  crashed processors).}
\end{figure}
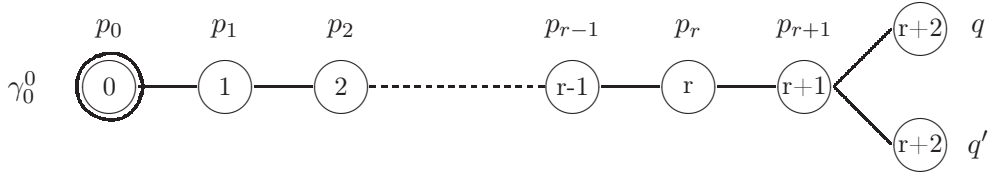

Let $\gamma_{0}^{0}$ be the following configuration: $\forall i\in \{0,\ldots,r+1\}, c_{p_{i}}=i$, $c_{q}=c_{q'}=r+2$ and $p_{0}$  crashed (see Figure \ref{fig:Figure6}). Note that, for any execution $\epsilon$ starting from $\gamma_{0}^{0}$, one of the processors $q$ and $q'$ must be enabled to modify its clock in a finite time	(otherwise the network would be starved following Lemma \ref{lem:blocage}). This implies the existence of a fragment of execution $\epsilon_{a}^{0}=\gamma_{0}^{0}\gamma_{1}^{0}\ldots\gamma_{k}^{0}$ with the following properties:

\begin{enumerate}
\item $k\geq 1$ if there exists $i\in \{0,\ldots,r+1\}$ such that $p_{i}$ is enabled in $\gamma_{0}^{0}$, $k=0$ otherwise;
\item $\epsilon_{a}^{0}$ contains no modification of clock values;
\item $\gamma_{k}^{0}$ is the first configuration where $q$ or $q'$ is enabled to modify its clock value.
\end{enumerate}

Assume now that the scheduling of $\epsilon_{a}^{0}$ satisfies the following property: at each step, the daemon chooses the processor that was last activated among enabled processors. Note that this scenario is compatible with a strongly fair scheduling.

Let us study the following cases:

\begin{description}
\item[Case 1:] $q$ is enabled in $\gamma_{k}^{0}$ for a modification of its clock value. The safety property of $\mathcal{A}$ implies that the value of $c_{q}$ should be modified  either to $r$ or to $r+1$.

\begin{description}
\item[Case 1.1:] The value of $c_{q}$ is modified to $r$.\\
Since $\mathcal{A}$ is a priority unison, there exists by definition a fragment of execution $\epsilon_{b1}^{0}=\gamma_{k}^{0}\gamma_{k+1}^{0}\ldots\gamma_{k+r}^{0}$ that contains only actions of $q$ such that (i) in the steps from $\gamma_{k}^{0}$ to $\gamma_{k+r-1}^{0}$ the clock value of $q$ is not modified and (ii) in the step	$\gamma^{0}_{k+r-1}\rightarrow\gamma^{0}_{k+r}$ the clock value of $q$ is incremented.

Since $\mathcal{A}$ is a priority unison, there exists by definition a fragment of execution $\epsilon_{b2}^{0}=\gamma_{k+r}^{0}\gamma_{k+r+1}^{0}\ldots\gamma_{k+j}^{0}$ that contains only executions of a rule by $q$ such that (i) in the steps from $\gamma_{k+r}^{0}$ to $\gamma_{k+j-1}^{0}$ the clock value of $q$ is not modified and (ii) in the step	$\gamma^{0}_{k+j-1}\rightarrow\gamma^{0}_{k+j}$ the clock value of $q$ is incremented.

Let  $\epsilon_{b}^{0}$ be $\epsilon_{b1}^{0}\epsilon_{b2}^{0}$.

\item[Case 1.2:] The value of $c_{q}$ is modified to $r+1$.\\
Since $\mathcal{A}$ is a priority unison, there exists by definition a fragment of execution $\epsilon_{b}^{0}=\gamma_{k}^{0}\gamma_{k+1}^{0}\ldots\gamma_{k+r}^{0}$ that contains only actions of $q$ such that (i) in the steps from $\gamma_{k}^{0}$ to $\gamma_{k+r-1}^{0}$ the clock value of $q$ is not modified and (ii) in the step	$\gamma^{0}_{k+r-1}\rightarrow\gamma^{0}_{k+r}$ the clock value of $q$ increments.
\end{description}

If $q'$ is enabled in the last configuration of $\epsilon_{b}^{0}$ \footnote{In this case, $q'$ was already enabled in the last configuration of $\epsilon_{a}^{0}$}, we can construct $\epsilon_{c}^{0}$ similarly to $\epsilon_{b}^{0}$ using processor $q'$. Otherwise, let $\epsilon_{c}^{0}$ be $\epsilon$ (the empty word).

\item[Case 2:] $q'$ is enabled in $\gamma_{k}^{0}$ for a modification of its clock value.\\
We can construct $\epsilon_{b}^{0}$ and $\epsilon_{c}^{0}$ similar to the case 1 by reversing the roles of $q$ and $q'$.
\end{description}

Let us define $\epsilon^{0}=\epsilon_{a}^{0}\epsilon_{b}^{0}\epsilon_{c}^{0}$. Notice that the clock values are identical in the first and the last configuration of $\epsilon^{0}$. This implies that we can infinitely repeat the previous reasoning in order to obtain an infinite execution $\epsilon=\epsilon^{0}\epsilon^{1}\ldots$ that satisfies:

\begin{itemize}
\item No correct processor is infinitely often enabled without executing a rule (since $q$ and $q'$ execute a rule infinitely often and others processors are chosen in function of their last execution of a rule, which implies that an infinitely often enabled processor executes a rule in a finite time). This execution satisfies a strongly fair scheduling.
\item The clock value of $p_{r+1}$ is never modified (whereas $d(p_{0},p_{r+1})=r+1$).
\end{itemize}

This execution contradicts the liveness property of $\mathcal{A}$, which implies the result.
\end{proof}

\section{A Universal Protocol for Chains and Rings}\label{sec:positive} 

In the following we consider the only remaining possibility results (see Table~\ref{table1}) that are related to asynchronous unison on chains and rings (\emph{i.e.} networks with a degree inferior to 3). In this section, we propose an $(1,0)-$FTSS algorithm for \textbf{AU} under a locally central strongly fair daemon. The proposed algorithm is both minimal and priority.

The main difference between our protocol and the many self-stabilizing unison algorithms existing in the literature~\cite{DH07cb,D97j,DW97j,PT97j} is that our correction rules use averaging rather than maximizing or minimizing, in order to not favor the clock value of a particular neighbor. Indeed, using a maximum or a minimum strategy could make the chosen neighbor prevent stabilization if it is crashed. The averaging idea was previously studied in \cite{LR04c} in a non-stabilizing fault-free setting. \cite{LM85j} uses also average to perform clock synchronization in a non-stabilizing Byzantine-tolerant system. The main difference with our approach is that authors of \cite{LM85j} reject values that are too far from others (in order to avoid values proposed by Byzantine neighbors). In our case, we cannot reject any value due to the arbitrary initial clock values and the small number of available values (as our protocol operates on chains or rings, each processor has at most two neighbors).

\begin{algorithm}
\caption{($\mathcal{UFTSS}$): universal $(1,0)$-FTSS \textbf{AU} for chains and rings.}\label{algo:uftss}
\begin{small}		
\textbf{Data:}\\
- $N_{p}$: set of neighbors of $p$.\\
\textbf{Variable:}\\
- $c_{p}$: natural integer representing the clock of the processor.\\
\textbf{Macros:}\\
- For $A\subseteq \mathbb{N}$ and $a\in \mathbb{N}$, $next(A,a)=\begin{cases}
a+1\;if\;a+1\in A\\
min\{A\}\;otherwise\end{cases}$.\\
- For $q\in N_{p}$, $poss(q)=\begin{cases}
\{c_{q}-1,c_{q},c_{q}+1\}\;if\;c_{q}\neq 0\\
\{c_{q},c_{q}+1\}\;otherwise\end{cases}$.\\
- $Inter(N_{p})=\underset{q\in N_{p}}{\bigcap}poss(q)$.\\
\textbf{Rules:}\\
\textbf{/{*}} Normal rule \textbf{{*}/}\\
$\boldsymbol{(N)} :: |Inter(N_{p})|\geq 2\longrightarrow c_{p}:=next\left(Inter(N_{p}),c_{p}\right)$\\
\textbf{/{*}} Correction rules \textbf{{*}/}\\
$\boldsymbol{(C_{1})} :: \left(|Inter(N_{p})|=0\right)\wedge\left(c_{p}\neq \left\lceil \frac{\underset{q\in N_{p}}{\sum}c_{q}}{|N_{p}|}\right\rceil\right)\wedge \left(c_{p}\neq \left\lfloor \frac{\underset{q\in N_{p}}{\sum}c_{q}}{|N_{p}|}\right\rfloor\right)\longrightarrow c_{p}:=\left\lfloor \frac{\underset{q\in N_{p}}{\sum}c_{q}}{|N_{p}|}\right\rfloor$\\
$\boldsymbol{(C_{2})} :: (Inter(N_{p})=\{h\})\wedge(c_{p}\neq h)\longrightarrow c_{p}:=h$
\end{small}		
\end{algorithm}

\subsection{Our Algorithm}

The main idea of our algorithm follows. Each processor checks if it is ``locally synchronized'', \emph{i.e.} if the drift between its clock value and the clock values of its neighbors does not exceed 1. If a processor $p$ is ``locally synchronized'', it modifies its clock value in a finite time in order to preserve this property. Otherwise, $p$ corrects its clock value in finite time. 

More precisely, each processor $p$ has only one variable: its clock denoted by $c_{p}$. At each step, every processor $p$ computes a set of \emph{possible clock values},	\emph{i.e.} the set of clock values that have a drift of at most 1 with respect to all neighbors of $p$ (note that computing this set relies only on the clock values of $p$'s neighbors, but not on the one of $p$). This set is denoted by $Inter(N_{p})$.

Then, the following cases may appear:\\
\indent - $|Inter(N_{p})|=0$, then $p$ has two neighbors and the drift between their clock values is strictly greater than 2. In this case, $p$ is enabled to take the average value between these two clock values if its clock does not have yet this value.\\
\indent - $|Inter(N_{p})|=1$, then $p$ has two neighbors and the drift between their clock values is exactly 2. In this case, $p$ is enabled	to take the average value between these two clock values if its clock does not have yet this value.\\
\indent - $|Inter(N_{p})|\geq 2$, then $p$ has one neighbor or the drift between the clock values of its two neighbors is strictly less than 2. In this case, $p$ is enabled to modify its clock value as follows: if $c_{p}+1\in Inter(N_{p})$, then $c_{p}$ is modified to $c_{p}+1$, otherwise $c_{p}$ is modified to $min\{Inter(N_{p})\}$.

The reader can find some examples of execution of our algorithm in Figures \ref{fig:Exemple1} to \ref{fig:Exemple4}.

The detailed description of our solution is proposed in Algorithm \ref{algo:uftss}.

\input{figure8_11}

\subsection{Correction Proof Road Map}

In this section, we present the key ideas in order to prove the correctness of our algorithm.

First, we introduce some useful notations:

\begin{notation}
Let $p$ be a processor. If $q$ denotes one of its neighbors, we denote the other neighbor by $\bar{q}$ (if this neighbor exists). 
\end{notation}

\begin{notation}
We denote the value of $c_{p}$ for a processor $p$ in a configuration $\gamma_{i}$ by $\left(c_{p}\right)^{\gamma_{i}}$.

We denote the value of $Inter(N_{p})$ for a processor $p$ in a configuration $\gamma_{i}$ by $\left(Inter(N_{p})\right)^{\gamma_{i}}$.
\end{notation}

In order to prove that $\mathcal{UFTSS}$ is a $(1,0)$-ftss algorithm for \textbf{AU} under a locally central strongly fair daemon on a chain and on a ring (see Proposition \ref{prop:ftss}), we prove in the sequel the following properties:

\begin{enumerate}
\item $\mathcal{UFTSS}$ is a self-stabilizing algorithm for \textbf{AU} under a locally central strongly fair daemon on a chain (Proposition \ref{prop:SSChaine}).
\item $\mathcal{UFTSS}$ is a self-stabilizing algorithm for \textbf{AU} under a locally central strongly fair daemon on a chain even if one processor is crashed in the initial configuration (Proposition \ref{prop:SSCrashChaine}).
\item $\mathcal{UFTSS}$ is a self-stabilizing algorithm for \textbf{AU} under a locally central strongly fair daemon on a ring (Proposition \ref{prop:SSAnneau}).
\item $\mathcal{UFTSS}$ is a self-stabilizing algorithm for \textbf{AU} under a locally central strongly fair daemon on a ring even if one processor is crashed in the initial configuration (Proposition \ref{prop:SSCrashAnneau}).
\end{enumerate}

The proof of each of these 4 propositions is deduced from 3 lemmas as follows:

\begin{enumerate}
\item Firstly, we prove that $\mathcal{UFTSS}$ satisfies the closure of the safety of \textbf{UAU} under the considered hypothesis (\emph{i.e.} if there exists a configuration $\gamma$ such that $\gamma\in\Gamma_{1}$, then every configuration $\gamma'$ reachable from $\gamma$ satisfies: $\gamma'\in\Gamma_{1}$, see respectively Lemma \ref{lem:clotureChaine}, \ref{lem:clotureCrashChaine}, \ref{lem:clotureAnneau}, and \ref{lem:clotureCrashAnneau}).

The idea of the proof is as follows: we first prove that only the normal rule is enabled in such a configuration and then, we show that this rule ensures the closure of the safety property.
\item Secondly, we prove that $\mathcal{UFTSS}$ satisfies liveness of \textbf{UAU} under the considered hypothesis in every execution starting from a legitimate configuration (\emph{i.e.} every (correct) processor increments infinitely often its clock, see respectively Lemma \ref{lem:vivaciteChaine}, \ref{lem:vivaciteCrashChaine}, \ref{lem:vivaciteAnneau}, and \ref{lem:vivaciteCrashAnneau}).

This proof is done in the following way: we first show that every (correct) processor executes infinitely often the normal rule in every execution starting from a configuration $\gamma\in\Gamma_{1}$ and then, we show that if a processor executes infinitely often the normal rule, it increments its clock in a finite time.
\item Finally, we prove that $\mathcal{UFTSS}$ converges to a legitimate configuration of \textbf{UAU} under the considered hypothesis in every execution (\emph{i.e.} there exists a configuration $\gamma\in\Gamma_{1}$ in every execution, see respectively Lemma \ref{lem:convergenceChaine}, \ref{lem:convergenceCrashChaine}, \ref{lem:convergenceAnneau}, and \ref{lem:convergenceCrashAnneau}).

In order to complete this proof we study a potential function.
\end{enumerate}

\subsection{Proof on a Chain}

In this section, we assume that our algorithm is executed on a chain under a strongly fair locally central daemon. In the following we prove that $\mathcal{UFTSS}$ is a FTSS \textbf{UAU} (that implies that it is a FTSS \textbf{AU}) under these assumptions. The proof contains two major steps:

\begin{itemize}
\item First, we prove that our algorithm is self-stabilizing.
\item Second, we  prove that our algorithm is self-stabilizing even if the initial configuration contains a crashed processor.
\end{itemize}

\subsubsection{Proof of Self-Stabilization}

In this section, $\epsilon=\gamma_{0},\gamma_{1}\ldots$ denotes an execution of $\mathcal{UFTSS}$ where there is no crash.

Firstly, we are going to prove the closure of our algorithm.

\begin{lemma}\label{lem:clotureChaine}
If there exists $i\geq 0$ such that $\gamma_{i}\in \Gamma_{1}$, then $\gamma_{i+1}\in\Gamma_{1}$.
\end{lemma}

\begin{proof}
Assume that there exists $i\geq 0$ such that $\gamma_{i}\in \Gamma_{1}$. This implies that $\forall p\in V,$ $\left(Inter(N_{p})\right)^{\gamma_{i}}\neq\emptyset$ and then the rule $\boldsymbol{(C_{1})}$ is not enabled in $\gamma_{i}$. Assume rule $\boldsymbol{(C_{2})}$ is enabled in $\gamma_{i}$. This implies that $\left(Inter(N_{p})\right)^{\gamma_{i}}=\{h\}$ and that $\left(c_{p}\right)^{\gamma_{i}}\neq h$. Then, we have $\gamma_{i}\notin \Gamma_{1}$ (since if $\left(c_{p}\right)^{\gamma_{i}}\neq h$, then the following holds: $\exists q\in N_{p},|\left(c_{p}\right)^{\gamma_{i}}-\left(c_{q}\right)^{\gamma_{i}}|\geq 2$). This contradiction allows us to conclude that the enabled processors in $\gamma_{i}$ are only enabled for rule $\boldsymbol{(N)}$.

Let $p$ be a processor that executes a rule during the step $\gamma_{i}\rightarrow\gamma_{i+1}$. Since the daemon is locally central, neighbors of $p$ do not execute a rule during this step (their clock values remain identical). Assume the following holds: $\exists q\in N_{p}, |\left(c_{p}\right)^{\gamma_{i+1}}-\left(c_{q}\right)^{\gamma_{i+1}}|\geq 2$. By construction of rule $\boldsymbol{(N)}$, $\left(c_{p}\right)^{\gamma_{i+1}}\in\left(Inter(N_{p})\right)^{\gamma_{i}}$. By construction, $\left(Inter(N_{p})\right)^{\gamma_{i}}\subseteq \{\left(c_{q}\right)^{\gamma_{i}}-1,\left(c_{q}\right)^{\gamma_{i}},\left(c_{q}\right)^{\gamma_{i}}+1\}$. It follows that $\forall q\in N_{p},|\left(c_{p}\right)^{\gamma_{i+1}}-\left(c_{q}\right)^{\gamma_{i+1}}|< 2$ for each processor $p$ that executes a rule (since $\forall q\in N_{p},\left(c_{q}\right)^{\gamma_{i}}=\left(c_{q}\right)^{\gamma_{i+1}}$). Overall, $\gamma_{i+1}\in \Gamma_{1}$.
\end{proof}

Secondly, we prove the liveness of our algorithm.

\begin{lemma}\label{lem:prelemVivaciteChaine}
$\forall\gamma_{0}\in\Gamma_{1},\forall p\in V,$ $p$ executes the rule $\boldsymbol{(N)}$ in a finite time in any execution starting from $\gamma_{0}$.
\end{lemma}

\begin{proof}
Let $\gamma\in\Gamma_{1}$. Following Lemma \ref{lem:clotureChaine}, the only  enabled rule is $\boldsymbol{(N)}$. We prove this property by induction. To this end, we define the following property (where $p$ denotes a processor):\\
$\boldsymbol{(P_{d})}$ : If $d$ is the distance between $p$ and the closest end of the chain, then $p$ executes the rule $\boldsymbol{(N)}$ in a finite time	in any execution starting from $\gamma_{0}$.

\begin{description}
\item[Initialization ($d=0$):] For all $\gamma'$, configurations contained in an execution starting from $\gamma_{0}$, $p$ is enabled for rule $\boldsymbol{(N)}$ since $\left(Inter(N_{p})\right)^{\gamma'}\supseteq\{\left(c_{q}\right)^{\gamma'},\left(c_{q}\right)^{\gamma'}+1\}$ where $q$ denotes the only neighbor of $p$.	Since the daemon is strongly fair, $p$ executes a rule in a finite time.

\item[Induction ($d>0$):] Assume $\boldsymbol{(P_{d-1})}$ is true. Denote by $q$ the neighbor of $p$ that is on the half-chain starting with $p$ of length $d$. Assume for the sake of contradiction that $p$ is never enabled for rule $\boldsymbol{(N)}$ in an execution $\epsilon$ starting from $\gamma_{0}\in\Gamma_{1}$. This implies that, for each configuration $\gamma'$ that is contained in $\epsilon$, we have $|\left(Inter(N_{p})\right)^{\gamma'}|=1$ (since if $|\left(Inter(N_{p})\right)^{\gamma'}|=0$, then $\gamma'\notin\Gamma_{1}$). Let us study the following cases (remind that, if $q$ denotes a neighbor of $p$, $\bar{q}$ denotes the second neighbor of $p$ as stated in Notation 1):

\begin{description}
\item[Case 1:] $\bar{q}$ never executes a rule in $\epsilon$ (this implies that $c_{\bar{q}}$ is a constant in $\epsilon$).\\
It follows that: $\forall\gamma'\in\epsilon, \left(c_{q}\right)^{\gamma'}=\left(c_{\bar{q}}\right)^{\gamma'}+2$ or	$\left(c_{q}\right)^{\gamma'}=\left(c_{\bar{q}}\right)^{\gamma'}-2$.

As $q$ executes infinitely often rule $\boldsymbol{(N)}$, its clock moves at each activation from a value to the other. Hence, we have $\left(c_{q}\right)^{\gamma'}=\left(c_{\bar{q}}\right)^{\gamma'}-2$ in a finite time. Then, the next activation of $q$ moves its clock value to $\left(c_{\bar{q}}\right)^{\gamma'}+2$, which is contradictory with the construction of macro $next$ (it can only increment the clock value by 1 or decrement it).

\item[Case 2:] $\bar{q}$ executes a rule in a finite time in $\epsilon$.\\
Let $\gamma\rightarrow\gamma'$ be the first step when $\bar{q}$ executes the rule $\boldsymbol{(N)}$. It is known that, for any $\gamma\in\Gamma_{1}$:
\[|\left(Inter(N_{p})\right)^{\gamma}|=1\Rightarrow\begin{cases}
\left(c_{\bar{q}}\right)^{\gamma}=(\left(c_{p}\right)^{\gamma}-1)\wedge\left(c_{q}\right)^{\gamma}=(\left(c_{p}\right)^{\gamma}+1)\:\boldsymbol{(A)}\\
or\\
\left(c_{\bar{q}}\right)^{\gamma}=(\left(c_{p}\right)^{\gamma}+1)\wedge\left(c_{q}\right)^{\gamma}=(\left(c_{p}\right)^{\gamma}-1)\:\boldsymbol{(B)}\\
\end{cases}\]

Let us study the following cases:
\begin{description}
\item[Case 2.1:] $\boldsymbol{(A)}$ is true in $\gamma$ and $\boldsymbol{(B)}$ is true in $\gamma'$.	The clock move of $\bar{q}$ is in contradiction with the construction of macro $next$.

\item[Case 2.2:] $\boldsymbol{(B)}$ is true in $\gamma$ and $\boldsymbol{(A)}$ is true in $\gamma'$.	The clock move of $q$ is in contradiction with the construction of macro $next$.
\end{description}

This proves that case 2 is contradictory.
\end{description}

Since the two cases are contradictory, we can conclude that $p$ is enabled for rule $\boldsymbol{(N)}$ in a finite time in every execution starting from a configuration $\gamma\in\Gamma_{1}$. Since the daemon is strongly fair, we can say that $p$ executes rule $\boldsymbol{(N)}$ in a finite time in every execution starting from $\gamma_{0}$. Consequently $\boldsymbol{(P_{d})}$ is true.
\end{description}
\end{proof}

The above property implies that  $\forall\gamma_{0}\in\Gamma_{1},\forall p\in V,$ $p$ executes the rule $\boldsymbol{(N)}$ infinitely often in every execution starting from $\gamma_{0}$.

\begin{lemma}\label{lem:vivaciteChaine}
If $\gamma\in\Gamma_{1}$, then any processor increments its clock in a finite time in any execution starting from $\gamma$.
\end{lemma}

\begin{proof}
Assume for the sake of contradiction that there exists a processor $p$ and an execution $\epsilon$ starting from $\gamma_{0}\in\Gamma_{1}$ such that $p$ never increments its clock in $\epsilon$.

Let $\alpha=\left(c_{p}\right)^{\gamma_{0}}$. By Lemma \ref{lem:prelemVivaciteChaine}, $p$ executes infinitely often $\boldsymbol{(N)}$. But, it never increments its clock, which implies that $next(\left(Inter(N_{p})\right)^{\gamma},\left(c_{p}\right)^{\gamma})=min\{\left(Inter(N_{p}\right)^{\gamma})\}$ at each execution of a rule by $p$ (in a configuration $\gamma$). Since $\forall\gamma\in\Gamma_{1},\forall q\in N_{p},|\left(c_{p}\right)^{\gamma}-\left(c_{q}\right)^{\gamma}|<2$ and $\forall q\in N_{p},\left(Inter(N_{p})\right)^{\gamma}\subseteq \{\left(c_{q}\right)^{\gamma}-1,\left(c_{q}\right)^{\gamma},\left(c_{q}\right)^{\gamma}+1\}$, we have:	$min\{\left(Inter(N_{p})\right)^{\gamma}\}\leq \left(c_{p}\right)^{\gamma}$.

Assume that there exists $\gamma\in\Gamma_{1}$ such that $min\{\left(Inter(N_{p})\right)^{\gamma}\}=\left(c_{p}\right)^{\gamma}$. This implies that there exists $q\in N_{p}$ such that $\left(c_{q}\right)^{\gamma}=\left(c_{p}\right)^{\gamma}+1$.

Remind that, if $q$ denotes a neighbor of $p$, $\bar{q}$ denotes the second neighbor of $p$ as stated in Notation 1. If $\bar{q}$ does not exist or if $\left(c_{\bar{q}}\right)^{\gamma}\in\{\left(c_{p}\right)^{\gamma},\left(c_{p}\right)^{\gamma}+1\}$, then  $\left(c_{p}\right)^{\gamma}+1\in\left(Inter(N_{p})\right)^{\gamma}$. This contradicts $next(\left(Inter(N_{p})\right)^{\gamma},\left(c_{p}\right)^{\gamma})=min\{\left(Inter(N_{p}\right)^{\gamma})\}$. We deduce that $\bar{q}$ exists and that $\left(c_{\bar{q}}\right)^{\gamma}=\left(c_{p}\right)^{\gamma}-1$. This implies that $\boldsymbol{(N)}$ is not enabled for $p$.

We can deduce that, if rule $\boldsymbol{(N)}$ is executed by a processor $p$ in a configuration $\gamma$, then  $min\{\left(Inter(N_{p})\right)^{\gamma}\}<\left(c_{p}\right)^{\gamma}$. We can now state that, in at most $\alpha$ executions of  $p$,  $c_{p}=0$. The next execution of $p$ increments its clock value, which contradicts the assumption on $p$ and the construction of $\epsilon$. Then, we obtain the result. 
\end{proof}

In the following we prove the convergence of our algorithm.

Let $\gamma\in\Gamma$, we define the following notations:
\[\begin{array}{c}\forall e=\{p,q\}\in E, \omega(e,\gamma)=|\left(c_{p}\right)^{\gamma}-\left(c_{q}\right)^{\gamma}|\\
\forall p\in V, \varpi(p,\gamma)=\underset{e\in E/p\in e}{max}\{\omega(e,\gamma)\}\\
\forall i\in\mathbb{N},p(i,\gamma)=|\{e\in E/\omega(e,\gamma)=i\}|\end{array}\]

Consider the following potential function:
\[P : \begin{cases} \Gamma\longrightarrow\mathbb{N}^{\infty}\\
\gamma\longmapsto\left(\ldots,0,0,p(k,\gamma),p(k-1,\gamma),\ldots,p(2,\gamma)\right)\:with\:k=\underset{e\in E}{max}\{\omega(e,\gamma)\}\end{cases}\]

To compare values of $P$, we define the following total order. If $\gamma$ and $\gamma'$ are two configurations such that $P(\gamma)=(\ldots,0,p_i,p_{i-1},\ldots,p_2)$ and $P(\gamma')=(\ldots,0,q_j,q_{j-1},\ldots,q_2)$, then 
\[P(\gamma)>P(\gamma')\Leftrightarrow\begin{cases}
i>j\\
\text{or}\\
(i=j)\wedge(\exists t\in\{2,\ldots,i\}, (\forall k\in\{t+1,\ldots,i\}, p_k=q_k)\wedge(p_k>q_k))
\end{cases}\]

The following properties are satisfied:
\[\begin{array}{c}\forall \gamma\in\Gamma, P(\gamma)\geq(\ldots0,0)\\
\forall \gamma\in\Gamma, \gamma\in\Gamma_{1}\Leftrightarrow P(\gamma)=(\ldots,0,0)\\
\forall \gamma\in\Gamma, \gamma\in\Gamma\setminus\Gamma_{1}\Leftrightarrow P(\gamma)>(\ldots,0,0)\end{array}\]

\begin{lemma}\label{lem:prelem1ConvergenceChaine}
If $\gamma\in\Gamma\setminus\Gamma_{1}$, then every step $\gamma\rightarrow\gamma'$, which contains the execution of a rule by a processor $p$ such that $\varpi(p)\geq 2$ satisfies $P(\gamma')<P(\gamma)$.
\end{lemma}

\begin{proof}
Let $\gamma\in\Gamma\setminus\Gamma_{1}$. Let $\gamma\rightarrow\gamma'$ be a step that contains the execution of a rule by a processor $p$ such that $\varpi(p)\geq 2$ and $\gamma\in\Gamma\setminus\Gamma_{1}$. Since the daemon is locally central, neighbors of $p$ do not modify their clocks during this step. Consider the following cases:

\begin{description}
\item[Case 1:] $p$'s degree equals $1$.\\
Let $q$ be its only neighbor and $j=\omega(\{p,q\},\gamma)=|\left(c_{p}\right)^{\gamma}-\left(c_{q}\right)^{\gamma}|$. $\left(Inter(N_{p})\right)^{\gamma}=\{\left(c_{q}\right)^{\gamma}-1,\left(c_{q}\right)^{\gamma},\left(c_{q}\right)^{\gamma}+1\}$. It follows that $p$ executed rule $\boldsymbol{(N)}$. So, we have $|\left(c_{p}\right)^{\gamma'}-\left(c_{q}\right)^{\gamma'}|\leq 1$. Then: $\varpi(\{p,q\},\gamma')\leq 1$ and :
\[\begin{array}{c}P(\gamma)=\left(\ldots,0,0,p(k,\gamma),p(k-1,\gamma),\ldots,p(j,\gamma),\ldots,p(2,\gamma)\right)\\
P(\gamma')=\left(\ldots,0,0,p(k,\gamma),p(k-1,\gamma),\ldots,p(j,\gamma)-1,\ldots,p(2,\gamma)\right)\end{array}\]
And then: $P(\gamma')<P(\gamma)$.

\item[Case 2:] $p$'s degree equals $2$.\\
Let $q$ be the neighbor of $p$ such that $\omega(\{p,q\},\gamma)=\varpi(p,\gamma) \geq 2$ and denote $j=\omega(\{p,\bar{q}\},\gamma)\leq \varpi(p,\gamma)$, $e=\{p,q\}$ and $\bar{e}=\{p,\bar{q}\}$. Consider the following cases:

\begin{description}
\item[Case 2.1:] $p$ executed the rule $\boldsymbol{(N)}$ during the step $\gamma\rightarrow\gamma'$.\\
By construction of $\left(Inter(N_{p})\right)^{\gamma}$, we have $\omega(e,\gamma')\leq 1$ and $\omega(\bar{e},\gamma')\leq 1$. Then:
\[\begin{array}{c}P(\gamma)=\left(\ldots,0,0,p(k,\gamma),p(k-1,\gamma),\ldots,p(\varpi(p,\gamma),\gamma),\ldots,p(j,\gamma),\ldots,p(2,\gamma)\right)\\
P(\gamma')=\left(\ldots,0,p(k,\gamma),\ldots,p(\varpi(p,\gamma),\gamma)-1,\ldots,p(j,\gamma)-1,\ldots,p(2,\gamma)\right)\end{array}\]
And then: $P(\gamma')<P(\gamma)$.

\item[Case 2.2:] $p$ executed the rule $\boldsymbol{(C_{2})}$ during the step $\gamma\rightarrow\gamma'$.\\
This case is similar to the case 2.1.

\item[Case 2.3:] $p$ executed the rule $\boldsymbol{(C_{1})}$ during the step $\gamma\rightarrow\gamma'$.\\
Let us study the following cases:

\begin{description}
\item[Case 2.3.1:] We have: $\left(c_{q}\right)^{\gamma}<\left(c_{\bar{q}}\right)^{\gamma}$.\\
By hypothesis, we know that $\omega(e,\gamma)\geq\omega(\bar{e},\gamma)$ and then: 
\[\left(c_{p}\right)^{\gamma}\geq\frac{\left(c_{q}\right)^{\gamma}+\left(c_{\bar{q}}\right)^{\gamma}}{2}\]

1) Assume that $\left(c_{p}\right)^{\gamma}>\left(c_{\bar{q}}\right)^{\gamma}+\frac{\left(c_{q}\right)^{\gamma}+\left(c_{\bar{q}}\right)^{\gamma}}{2}$.

We can say that: 
\[\begin{array}{c}\omega(e,\gamma)>\left(c_{\bar{q}}\right)^{\gamma}-\left(c_{q}\right)^{\gamma}+\frac{\left(c_{q}\right)^{\gamma}+\left(c_{\bar{q}}\right)^{\gamma}}{2}\\
\omega(e,\gamma')=\left\lfloor\frac{\left(c_{q}\right)^{\gamma}+\left(c_{\bar{q}}\right)^{\gamma}}{2}\right\rfloor\end{array}\]
Then: $\omega(e,\gamma')<\omega(e,\gamma)$. 

On the other hand,
\[\begin{array}{c}\omega(\bar{e},\gamma)>\frac{\left(c_{q}\right)^{\gamma}+\left(c_{\bar{q}}\right)^{\gamma}}{2}\\
\omega(\bar{e},\gamma')=\left(c_{\bar{q}}\right)^{\gamma}-\left\lfloor\frac{\left(c_{q}\right)^{\gamma}+\left(c_{\bar{q}}\right)^{\gamma}}{2}\right\rfloor\end{array}\]
Then: $\omega(\bar{e},\gamma')\leq\omega(\bar{e},\gamma)$. 

In conclusion, we have: $P(\gamma')<P(\gamma)$.

2) Assume that $\left(c_{p}\right)^{\gamma}\leq\left(c_{\bar{q}}\right)^{\gamma}+\frac{\left(c_{q}\right)^{\gamma}+\left(c_{\bar{q}}\right)^{\gamma}}{2}$.

We have then: 
\[\begin{array}{c}\omega(e,\gamma)>\frac{\left(c_{q}\right)^{\gamma}+\left(c_{\bar{q}}\right)^{\gamma}}{2}\\
\omega(e,\gamma')=\left\lfloor\frac{\left(c_{q}\right)^{\gamma}+\left(c_{\bar{q}}\right)^{\gamma}}{2}\right\rfloor\end{array}\]
Then: $\omega(e,\gamma')<\omega(e,\gamma)$. 

In contrast, we have that: $\omega(\bar{e},\gamma')\geq\omega(\bar{e},\gamma)$. But we can say that $\omega(\bar{e},\gamma')<\omega(e,\gamma)$ (obvious if 
$\left(c_{p}\right)^{\gamma}>\left(c_{\bar{q}}\right)^{\gamma}$, due to the fact that
$\left(c_{p}\right)^{\gamma}>\left\lceil\frac{\left(c_{q}\right)^{\gamma}+\left(c_{\bar{q}}\right)^{\gamma}}{2}\right\rceil$ in the contrary case).

In conclusion, we have: $P(\gamma')<P(\gamma)$.
\item[Case 2.3.2:] We have $\left(c_{q}\right)^{\gamma}>\left(c_{\bar{q}}\right)^{\gamma}$.\\
This case is similar to the case 2.3.1 when we permute $q$ and $\bar{q}$.

\end{description}
\end{description}
\end{description}

That proves the result.
\end{proof}
			
\begin{lemma}\label{lem:prelem2ConvergenceChaine}	
If $\gamma_{0}\in\Gamma\setminus\Gamma_{1}$, then every execution starting from $\gamma_{0}$ contains the execution of a rule by a processor $p$ such that $\varpi(p,\gamma_{0})\geq 2$.
\end{lemma}

\begin{proof}
Let $\gamma_{0}\in\Gamma\setminus\Gamma_{1}$. We prove the result by contradiction. Assume that there exists an execution $\epsilon=\gamma_{0}\gamma_{1}\ldots$ starting from $\gamma_{0}$, which contains no execution of a rule by processors $p$ satisfying $\varpi(p,\gamma_{0})\geq 2$.

In a first time, assume that one end of the chain (denote it by $p$) satisfies: $\varpi(p,\gamma_{0})\geq 2$. Denote $q$ the only neighbor of $p$. If $q$ is activated during $\epsilon$, we obtain a contradiction (since $\varpi(q,\gamma_{0})\geq\varpi(p,\gamma_{0})\geq 2$). If $q$ is not activated during $\epsilon$, we obtain that $\forall i\in\mathbb{N},\left(Inter(N_{p})\right)^{\gamma_{i}}=\{\left(c_{q}\right)^{\gamma_{0}}-1,\left(c_{q}\right)^{\gamma_{0}},\left(c_{q}\right)^{\gamma_{0}}+1\}$, $p$ is so always enabled for rule $\boldsymbol{(N)}$. Since the daemon is strongly fair, $p$ executes a rule in a finite time, which is contradictory. We can deduce that the two ends of the chain satisfy: $\varpi(p,\gamma_{0})<2$.

Under a strongly fair daemon, the only way for a processor to never execute a rule is to be never enabled from a given configuration. Here, we assume that all processors $p$ satisfying $\varpi(p,\gamma_{0})\geq 2$ never execute a rule, which implies that the network satisfies:
\[\exists k\in\mathbb{N},\forall j\geq k,\forall p\in V/\varpi(p,\gamma_{0})\geq 2,\begin{cases} \left(Inter(N_{p})\right)^{\gamma_{j}}=\emptyset\\
and\\
\left(c_{p}\right)^{\gamma_{j}}\in\left\{\left\lceil\frac{\left(c_{q}\right)^{\gamma_{j}}+\left(c_{\bar{q}}\right)^{\gamma_{j}}}{2}\right\rceil,
\left\lfloor\frac{\left(c_{q}\right)^{\gamma_{j}}+\left(c_{\bar{q}}\right)^{\gamma_{j}}}{2}\right\rfloor\right\}\end{cases}\]

Number processors of the chain from $p_{1}$ to $p_{n}$. Let $i$ be the smallest integer such that $\varpi(p_{i},\gamma_{k})\geq 2$ (remark that, by hypothesis, $p_{i+1}$ never execute a rule, which implies that its clock value never changes). All these constraints allows us to say:
\[\begin{cases}
\left(c_{p_{i-1}}\right)^{\gamma_{k}}=\left(c_{p_{i}}\right)^{\gamma_{k}}+1\wedge\left(c_{p_{i+1}}\right)^{\gamma_{k}}=\left(c_{p_{i}}\right)^{\gamma_{k}}-2\:\boldsymbol{(A)}\\
or\\		\left(c_{p_{i-1}}\right)^{\gamma_{k}}=\left(c_{p_{i}}\right)^{\gamma_{k}}-1\wedge\left(c_{p_{i+1}}\right)^{\gamma_{k}}=\left(c_{p_{i}}\right)^{\gamma_{k}}+2\:\boldsymbol{(B)}\end{cases}\]

By a reasoning similar to these of the proof of Lemma \ref{lem:vivaciteChaine}, we can prove that all processors between $p_{0}$ and $p_{i-1}$ executes infinitely often the rule $\boldsymbol{(N)}$ in every execution starting from $\gamma_{k}$ even if $p_{i}$ never executes a rule (this is the case by hypothesis). By a reasoning similar to the one of the proof of Lemma \ref{lem:vivaciteChaine}, we can state that $c_{p_{i-1}}$ not remains constant. The construction of $Inter(N_{p_{i-1}})$ implies that  $\left(Inter(N_{p_{i-1}})\right)^{\gamma_{j}}\subseteq\{\left(c_{p_{i}}\right)^{\gamma_{k}}-1,\left(c_{p_{i}}\right)^{\gamma_{k}},\left(c_{p_{i}}\right)^{\gamma_{k}}+1\}$ for each $j\geq k$ (since $c_{p_{i}}$ does not change by hypothesis).

If we are in case $\boldsymbol{(A)}$, we can deduce that $c_{p_{i-1}}$ takes infinitely often the value $\left(c_{p_{i}}\right)^{\gamma_{k}}-1$ or $\left(c_{p_{i}}\right)^{\gamma_{k}}$. We can see that $p_{i}$ is enabled by $\boldsymbol{(N)}$ and $\boldsymbol{(C_{1})}$ respectively. This contradicts the construction of $k$ (recall that $p_{i}$ is never enabled in $\epsilon$ from $\gamma_{k}$).

If we are in case $\boldsymbol{(B)}$, we can deduce that $c_{p_{i-1}}$  takes infinitely often the value $\left(c_{p_{i}}\right)^{\gamma_{k}}+1$ or $\left(c_{p_{i}}\right)^{\gamma_{k}}$. We can see that $p_{i}$ is enabled by $\boldsymbol{(N)}$ and $\boldsymbol{(C_{1})}$ respectively. This contradicts the construction of $k$ (recall that $p_{i}$ is never enabled in $\epsilon$ from $\gamma_{k}$).
				
This finishes the proof.
\end{proof}

\begin{lemma}\label{lem:convergenceChaine}
There exists $i\geq 0$ such that $\gamma_{i}\in\Gamma_{1}$.
\end{lemma}

\begin{proof}
The result follows directly from Lemmas \ref{lem:prelem1ConvergenceChaine} and \ref{lem:prelem2ConvergenceChaine}.
\end{proof}

Finally, we can conclude:

\begin{proposition}\label{prop:SSChaine}
$\mathcal{UFTSS}$ is a self-stabilizing \textbf{AU} under a locally central strongly fair daemon.
\end{proposition}

\begin{proof}
Lemmas \ref{lem:clotureChaine}, \ref{lem:vivaciteChaine}, and \ref{lem:convergenceChaine} allows us to say that $\mathcal{UFTSS}$ is a self-stabilizing \textbf{UAU} under a locally central strongly fair daemon. Then, we can deduce the result.
\end{proof}

\subsubsection{Proof of Self-Stabilization in spite of a Crash}
	
In this section, $\epsilon=\gamma_{0},\gamma_{1}\ldots$ denotes an execution of $\mathcal{UFTSS}$ such that a processor $c$ is crashed in $\gamma_{0}$.

Firstly, we are going to prove the closure of our algorithm under these assumptions.

\begin{lemma}\label{lem:clotureCrashChaine}
If there exists $i\geq 0$ such that $\gamma_{i}\in \Gamma_{1}$, then $\gamma_{i+1}\in\Gamma_{1}$.
\end{lemma}

\begin{proof}
We can repeat the reasoning of Lemma \ref{lem:clotureChaine} since the fact that a processor is crashed or not does not modify the proof.
\end{proof}

Secondly, we are going to prove the liveness of our algorithm under these assumptions.

\begin{lemma}\label{lem:vivaciteCrashChaine}
If $\gamma_{0}\in\Gamma_{1}$, then every processor $p\neq c$ increments its clock in a finite time in $\epsilon$.
\end{lemma}

\begin{proof}
We repeat the reasoning of Lemma \ref{lem:vivaciteChaine} taking in account a processor $p\in V^{*}$.

In order to prove the property of Lemma \ref{lem:prelemVivaciteChaine}, we take $d$ as the distance between $p$ and the end $e$ of the chain that satisfy: no processor between $p$ and $e$ is crashed. This implies that the processor $q$ is not crashed. The case where $\bar{q}$ is crashed appear in the case 1 of the induction.

We can repeat the reasoning of the proof of Lemma \ref{lem:vivaciteChaine} since the fact that a processor is crashed or not does not modify the proof.
\end{proof}
			
Now, we are going to prove the convergence of our algorithm under these assumptions.

\begin{lemma}\label{lem:convergenceCrashChaine}
There exists $i\geq 0$ such that $\gamma_{i}\in\Gamma_{1}$.
\end{lemma}

\begin{proof}
We repeat the reasoning of Lemma \ref{lem:convergenceChaine} taking in account a processor $p\in V^{*}$.

We can repeat the reasoning of the proof of the property of Lemma \ref{lem:prelem1ConvergenceChaine} since the fact that a processor is crashed or not does not modify the proof.

In order to prove the property of Lemma \ref{lem:prelem2ConvergenceChaine}, we take a numbering of processors that ensures the following property: no processor between $p_{0}$ and $p_{i}$ (including) is crashed. It is always possible to choose such numbering since there exists at least one edge $e$ such that $\omega(e,\gamma_{k})\geq 2$ by hypothesis, which implies that there exists at least two processors $p$ such that $\varpi(p,\gamma_{k})\geq 2$, which allows us to choose one that is not crashed. The case when $p_{i+1}$ is crashed does not modify the proof since we assumed that this processor never executes a rule.
\end{proof}

Finally, we can conclude:

\begin{proposition}\label{prop:SSCrashChaine}
$\mathcal{UFTSS}$ is a self-stabilizing \textbf{AU} under a locally central strongly fair daemon even if a processor is crashed in the initial configuration.
\end{proposition}

\begin{proof}
Lemmas \ref{lem:clotureCrashChaine}, \ref{lem:vivaciteCrashChaine}, and \ref{lem:convergenceCrashChaine} allows us to say that $\mathcal{UFTSS}$ is a self-stabilizing \textbf{UAU} under a locally central strongly fair daemon even if a processor is crashed in the initial configuration. Then, we can deduce the result.
\end{proof}
			
\subsection{Proof on a Ring}
  
In this section, we assume that our algorithm is executed on a ring under a strongly fair locally central daemon. In fact, we are going to show that $\mathcal{UFTSS}$ is a FTSS \textbf{UAU} (that implies that it is a FTSS \textbf{AU}) under these assumptions. The proof contains two major steps:

\begin{itemize}
\item Firstly, we show that our algorithm is self-stabilizing under these assumptions.
\item Secondly, we show that our algorithm is self-stabilizing even if the initial configuration contains a crashed processor under these assumptions.
\end{itemize}

\subsubsection{Proof of Self-Stabilization}

In this section, $\epsilon=\gamma_{0},\gamma_{1}\ldots$ denotes an execution of $\mathcal{UFTSS}$ where there is no crash.

Firstly, we are going to prove the closure of our algorithm under these assumptions.

\begin{lemma}\label{lem:clotureAnneau}
If there exists $i\geq 0$ such that $\gamma_{i}\in \Gamma_{1}$, then $\gamma_{i+1}\in\Gamma_{1}$.
\end{lemma}

\begin{proof}
We can repeat the reasoning of the proof of Lemma \ref{lem:clotureChaine} since the topology of the network has no impact on the proof.
\end{proof}
		
Secondly, we are going to prove the liveness of our algorithm under these assumptions.

\begin{lemma}\label{lem:prelemVivaciteAnneau}
$\forall\gamma_{0}\in\Gamma_{1},\forall p\in V,$ $p$ executes rule $\boldsymbol{(N)}$ in a finite time in every execution starting from $\gamma_{0}$.
\end{lemma}

\begin{proof}
Let $\gamma_{0}\in\Gamma_{1}$ (we have seen in the proof of Lemma \ref{lem:clotureChaine} that implies that only rule $\boldsymbol{(N)}$ can be enabled). Assume that there exists a processor $p$ and an execution $\epsilon=\gamma_{0},\gamma_{1}\ldots$ starting from $\gamma_{0}$ such that $p$ never execute a rule in $\epsilon$. Since the daemon is strongly fair, which implies that $\exists k\in\mathbb{N},\forall j\geq k$, $p$ is not enabled in $\gamma_{j}$

Since Processor $p$ is not enabled, it satisfies: $\exists q\in N_{p},\left(c_{p}\right)^{\gamma_{j}}=\left(c_{q}\right)^{\gamma_{j}}+1$ and $\left(c_{p}\right)^{\gamma_{j}}=\left(c_{\bar{q}}\right)^{\gamma_{j}}-1$. Let $i$ be the smallest integer greater than $k$ such that the step $\gamma_{i}\rightarrow\gamma_{i+1}$ contains the execution of rule by at least one neighbor of $p$. Let us study the following cases:

\begin{description}
\item[Case 1:] $q$ and $\bar{q}$ simultaneously execute a rule during the step $\gamma_{i}\rightarrow\gamma_{i+1}$.\\
Since $p$ is not enabled in $\gamma_{i+1}$ (by hypothesis) and that the execution of rule $\boldsymbol{(N)}$ always modifies the clock values (\emph{cf.}	proof of Lemma \ref{lem:vivaciteChaine}),we have: 
\[\begin{cases}\left(c_{p}\right)^{\gamma_{i}}=\left(c_{q}\right)^{\gamma_{i}}+1$ and $\left(c_{p}\right)^{\gamma_{i}}=\left(c_{\bar{q}}\right)^{\gamma_{i}}-1\\
and\\
\left(c_{p}\right)^{\gamma_{i+1}}=\left(c_{q}\right)^{\gamma_{i+1}}-1$ and $\left(c_{p}\right)^{\gamma_{i+1}}=\left(c_{\bar{q}}\right)^{\gamma_{i+1}}+1\end{cases}\]
The clock move of $\bar{q}$ contradicts the construction of rule $\boldsymbol{(N)}$ and $\left(Inter(N_{p})\right)^{\gamma_{i}}$. Therefore, this case is impossible.

\item[Case 2:] Only $q$ executes a rule during the step $\gamma_{i}\rightarrow\gamma_{i+1}$.\\
By construction of rule $\boldsymbol{(N)}$, $\left(Inter(N_{q})\right)^{\gamma_{i}}$, and the fact that the execution of this rule must change the clock value, we have: $\left(c_{q}\right)^{\gamma_{i+1}}\in\{\left(c_{p}\right)^{\gamma_{i}},\left(c_{p}\right)^{\gamma_{i}}-1\}$. Processor $p$ is then enabled for rule $\boldsymbol{(N)}$ (since the clocks of $p$ and $\bar{q}$ have not changed by hypothesis). This contradicts the construction of $k$. Therefore, this case is impossible.

\item[Case 3:] Only $ \bar{q}$ executes a rule during the step $\gamma_{i}\rightarrow\gamma_{i+1}$.\\
This case is similar to case 2.

\item[Case 4:] Neither $q$ nor $\bar{q}$ executes a rule during the step $\gamma_{i}\rightarrow\gamma_{i+1}$.\\
By the three previous contradictions, it is the only possible case.
\end{description}

We can deduce that $\forall j\geq k$, $q$ and $\bar{q}$ do not execute a rule in $\gamma_{j}$, which implies that their clock values remains constant from $\gamma_{k}$. If we repeat the previous reasoning, we obtain that it is possible only if the second neighbor of $q$ has a clock value equal to $\left(c_{p}\right)^{\gamma_{k}}+2$ and if the second neighbor of $\bar{q}$ have a clock value equals to $\left(c_{p}\right)^{\gamma_{k}}-2$, etc..

Since the ring has a finite length $n$, we obtain (following the same reasoning) that there exists two neighboring processors $p_{1}$ and $p_{2}$ such that  $\left(c_{p_{1}}\right)^{\gamma_{k}}=\left(c_{p}\right)^{\gamma_{k}}+\alpha$ and $\left(c_{p_{2}}\right)^{\gamma_{k}}=\left(c_{p}\right)^{\gamma_{k}}-\beta$ (with $\alpha$  and $\beta$ integers greater or equal to $1$ depending on the parity of $n$). Therefore, $|\left(c_{p_{1}}\right)^{\gamma_{k}}-\left(c_{p_{2}}\right)^{\gamma_{k}}|=\alpha+\beta\geq 2$. Then, we obtain that $\gamma_{k}\notin\Gamma_{1}$, which contradicts Lemma \ref{lem:clotureAnneau} and proves the lemma.
\end{proof}

\begin{lemma}\label{lem:vivaciteAnneau}
If $\gamma_{0}\in\Gamma_{1}$, then every processor increments its clock in a finite time in $\epsilon$.
\end{lemma}

\begin{proof}
The proof is similar to the one of Lemma \ref{lem:vivaciteChaine} using Lemma \ref{lem:prelemVivaciteAnneau} (instead of Lemma \ref{lem:prelemVivaciteChaine}) since the topology of the network has no impact on the proof.
\end{proof}
		
Now, we are going to prove the convergence of our algorithm under these assumptions.

In the following, we consider the potential function $P$ previously defined and use similar arguments as for the proof of Lemma \ref{lem:convergenceChaine}.

\begin{lemma}\label{lem:prelem1ConvergenceAnneau}
If $\gamma\in\Gamma\setminus\Gamma_{1}$, then every step $\gamma\rightarrow\gamma'$ that contains the execution of a rule of a processor $p$ such that $\varpi(p)\geq 2$ satisfies $P(\gamma')<P(\gamma)$.
\end{lemma}

\begin{proof}
The proof is similar to the proof of Lemma \ref{lem:prelem1ConvergenceChaine} since the topology of the network has no impact on the proof (note that the case 1 is impossible on a ring).
\end{proof}

\begin{lemma}\label{lem:prelem2ConvergenceAnneau}
If $\gamma_{0}\in\Gamma\setminus\Gamma_{1}$, then every execution starting from $\gamma_{0}$ contains the execution of a rule of a processor $p$ such that $\varpi(p,\gamma_{0})\geq 2$.
\end{lemma}

\begin{proof}
Let $\gamma_{0}\in\Gamma\setminus\Gamma_{1}$. Assume, for the sake of contradiction, that there exists an execution $\epsilon=\gamma_{0}\gamma_{1}\ldots$ starting from $\gamma_{0}$ that contains no execution of a rule by any processor	$p$ that satisfies $\varpi(p,\gamma_{0})\geq 2$. Since the daemon is strongly fair, this implies that $\exists k\in\mathbb{N},\forall j\geq k$, $p$ is not enabled in $\gamma_{j}$

Let $q$ be the neighbor of $p$ satisfying $\omega(\{p,q\},\gamma_{k})=\varpi(p,\gamma_{k})$. By hypothesis, $q$ never executes a rule. Therefore, its clock value remains  constant. Let us study the following cases:

\begin{description}
\item[Case 1:] $|\left(c_{q}\right)^{\gamma_{j}}-\left(c_{\bar{q}}\right)^{\gamma_{j}}|\leq 1$\\
It follows that $p$ is enabled for the rule $\boldsymbol{(N)}$ since $|\left(Inter(N_{p})\right)^{\gamma_{j}}|\geq 2$. This contradicts the construction of $k$.

\item[Case 2:] $|\left(c_{q}\right)^{\gamma_{j}}-\left(c_{\bar{q}}\right)^{\gamma_{j}}|=2$\\
It follows that $p$ is enabled for the rule $\boldsymbol{(C_{1})}$ since $\left(Inter(N_{p})\right)^{\gamma_{j}}=\{h\}$ and $\left(c_{p}\right)^{\gamma_{j}}\neq h$ (because $\varpi(p,\gamma_{j})=\varpi(p,\gamma_{k})\geq 2$). This contradicts the construction of $k$.

\item[Case 3:] $|\left(c_{q}\right)^{\gamma_{j}}-\left(c_{\bar{q}}\right)^{\gamma_{j}}|\geq 3$\\
By the two previous contradictions, it is the only possible case. Since $p$ is not enabled (by hypothesis), we obtain that:
\[\forall j\geq k,\begin{cases} \left(Inter(N_{p})\right)^{\gamma_{j}}=\emptyset\\
and\\
\left(c_{p}\right)^{\gamma_{j}}\in\left\{\left\lceil\frac{\left(c_{q}\right)^{\gamma_{j}}+\left(c_{\bar{q}}\right)^{\gamma_{j}}}{2}\right\rceil,\left\lfloor\frac{\left(c_{q}\right)^{\gamma_{j}}+\left(c_{\bar{q}}\right)^{\gamma_{j}}}{2}\right\rfloor\right\}\end{cases}\]
Since the clock values of $p$ and $q$ are constants by hypothesis, we can deduce that the one of $\bar{q}$ remains also constant (because, in the contrary case, $p$ becomes enabled, which contradicts the hypothesis). It follows: $\left(c_{q}\right)^{\gamma_{j}}<\left(c_{p}\right)^{\gamma_{j}}<\left(c_{\bar{q}}\right)^{\gamma_{j}}$ or $\left(c_{q}\right)^{\gamma_{j}}>\left(c_{p}\right)^{\gamma_{j}}>\left(c_{\bar{q}}\right)^{\gamma_{j}}$.
\end{description}
			
Since this reasoning holds for every processor on the ring, we can always label the nodes of any ring by $p_{0}$, $p_{1}$,\ldots,$p_{n}$ such that the following property is satisfied : $c_{p_{0}}<c_{p_{1}}<\ldots<c_{p_{n}}$.

But, the previous reasoning for processor $c_{p_{0}}$ implies that we have: $c_{p_{n}}<c_{p_{0}}<c_{p_{1}}$. It is impossible to satisfy simultaneously these two inequalities, which proves the lemma.
\end{proof}

\begin{lemma}\label{lem:convergenceAnneau}
There exists $i\geq 0$ such that $\gamma_{i}\in\Gamma_{1}$.
\end{lemma}

\begin{proof}
The result follows directly from Lemmas \ref{lem:prelem1ConvergenceAnneau} and \ref{lem:prelem2ConvergenceAnneau}.
\end{proof}
	
Finally, we can conclude:
		
\begin{proposition}\label{prop:SSAnneau}
$\mathcal{UFTSS}$ is a self-stabilizing \textbf{AU} under a locally central strongly fair daemon.
\end{proposition}

\begin{proof}
Lemmas \ref{lem:clotureAnneau}, \ref{lem:vivaciteAnneau}, and \ref{lem:convergenceAnneau} lead to the conclusion that $\mathcal{UFTSS}$ is a self-stabilizing \textbf{UAU} under a locally central strongly fair daemon.
\end{proof}
		
\subsubsection{Proof of Self-Stabilization in spite of a Crash}
	
In this section, $\epsilon=\gamma_{0},\gamma_{1}\ldots$ denotes an execution of $\mathcal{UFTSS}$ such that a processor $c$ is crashed in $\gamma_{0}$.

First, we prove the closure of our algorithm, then we prove the convergence property.

\begin{lemma}\label{lem:clotureCrashAnneau}
If there exists $i\geq 0$ such that $\gamma_{i}\in \Gamma_{1}$, then $\gamma_{i+1}\in\Gamma_{1}$.
\end{lemma}

\begin{proof}
This proof is similar to the proof of Lemma \ref{lem:clotureAnneau} since the fact that a processor is crashed or not does not modify the proof.
\end{proof}

Secondly, we are going to prove the liveness of our algorithm under these assumptions.

\begin{lemma}\label{lem:vivaciteCrashAnneau}
If $\gamma_{0}\in\Gamma_{1}$, then every processor $p\neq c$ increments its clock in a finite time in $\epsilon$.
\end{lemma}

\begin{proof}
This proof is similar to the proof of Lemma \ref{lem:vivaciteAnneau}.
\end{proof}
		
In the following we prove the convergence of our algorithm.

\begin{lemma}\label{lem:convergenceCrashAnneau}
There exists $i\geq 0$ such that $\gamma_{i}\in\Gamma_{1}$.
\end{lemma}

\begin{proof}
This proof is similar to the proof of Lemma \ref{lem:convergenceAnneau} since the fact that a processor is crashed or not does not modify the proof.
\end{proof}

Finally, we can conclude:

\begin{proposition}\label{prop:SSCrashAnneau}
$\mathcal{UFTSS}$ is a self-stabilizing \textbf{AU} under a locally central strongly fair daemon even if a processor is crashed in the initial configuration.
\end{proposition}

\begin{proof}
Lemmas \ref{lem:clotureCrashAnneau}, \ref{lem:vivaciteCrashAnneau}, and \ref{lem:convergenceCrashAnneau} allows us to say that $\mathcal{UFTSS}$ is a self-stabilizing \textbf{UAU} under a locally central strongly fair daemon even if a processor is crashed in the initial configuration. Then, we can deduce the result.
\end{proof}
		
\subsection{Conclusion}

We are now in position to state our final result:

\begin{proposition}\label{prop:ftss}
$\mathcal{UFTSS}$ is a $(0,1)$-ftss \textbf{AU} on a chain or a ring under a locally central strongly fair daemon.
\end{proposition}

\begin{proof}
This a direct consequence of Propositions \ref{prop:SSChaine}, \ref{prop:SSCrashChaine}, \ref{prop:SSAnneau}, and \ref{prop:SSCrashAnneau}.
\end{proof}

\section{Concluding Remarks}\label{sec:conclusion}

We presented the first study of FTSS protocols for dynamic tasks in asynchronous systems, and showed the intrinsic problems that are induced by the wide range of faults that we address. The combination of asynchrony and maintenance of liveness properties implies many impossibility results, and the deterministic protocol that we provided for one of the few remaining cases is optimal with respect to all impossibility results and containment measures. Then, we can observe that the results remain even if the weakly synchronized configuration definition is relaxed to allow neighbor clocks to be at most $\kappa$ away from each other, for some constant $\kappa$.

\paragraph{Generalization: $\kappa$-asynchronous unison.} In this paragraph, we briefly explain how to generalize the above results to a weaker problem. Assume that $\kappa\in\mathbb{N}^*$. In the $\kappa$-asynchronous unison problem ($\kappa$-\textbf{AU}), a drift of at most $\kappa$ units is allowed between clocks of any two neighbors. Hence, the \textbf{AU} problem corresponds to the $1$-\textbf{AU}.

Let us observe that a similar result to Lemma \ref{lem:blocage} holds in the case of $\kappa$-\textbf{AU}:

\begin{lemma}
Let $\mathcal{A}$ be a universal $(f,r)-$FTSS algorithm for $\kappa$-\textbf{AU} (under an asynchronous daemon). Let $\gamma$ be a configuration where a processor $p$ with $c_{p}\geq \kappa$ has two neighbors $q$ and $q'$ such that: $c_{q}=c_{p}-\kappa$ and $c_{q'}=c_{p}+\kappa$.
	If $p$ executes an action of $\mathcal{A}$ during the step $\gamma\longrightarrow\gamma'$, then this action does not modify the value of $c_{p}$.
  If $\mathcal{A}$ is also minimal, then the processor $p$ is not enabled for $\mathcal{A}$ in $\gamma$.
\end{lemma}

As Lemma \ref{lem:blocage} is the basis of proofs of Section \ref{sec:negative}, we can deduce that all impossibility results presented in Section \ref{sec:negative} still hold in the case of $\kappa$-\textbf{AU} .

In order to solve the $\kappa$-\textbf{AU} problem in the remaining cases, we modify Algorithm $\mathcal{UFTSS}$ (see Section \ref{sec:positive}) in the definition of macro $poss(q)$ in the following way:
\[\forall q\in N_{p}, poss(q)=\left\{max\{c_q-\kappa,0\},max\{c_q-\kappa,0\}+1,\ldots,c_q,\ldots,c_q+\kappa-1,c_q+\kappa\right\}\]

This modified algorithm is a universal $(0,1)$-FTSS $\kappa$-\textbf{AU} under a locally central strongly fair daemon on a chain or a ring (the proof is a simple generalization of the correctness proof of Section \ref{sec:positive}).

\paragraph{Open questions.} An immediate future work is to generalize the possibility result (that assumes a central scheduler) to cope with a distributed scheduler, or extend the impossibility proof in that case. There also remains the open case of protocols that neither satisfy the minimality or the priority properties (see Table~\ref{table1}). We conjecture that at least one of those properties is necessary for the purpose of \emph{deterministic} self-stabilization, yet none of those could be required for deterministic \emph{weak} stabilization~\cite{G01cb} (weak stabilization is a weaker property than self-stabilization since \emph{existence} of execution reaching a legitimate configuration is guaranteed). As recent results~\cite{DTY08c} hint that weak-stabilizing solutions can be easily turned into \emph{probabilistic} self-stabilizing ones, this raises the open question of the possibility of \emph{probabilistic} FTSS for dynamic tasks in asynchronous systems.

Another possible extension of our work is the feasibility of FTSS solutions for other reactive tasks, such as \emph{dining philosophers} and \emph{mutual exclusion}. In the case of dining philosophers, \cite{NA02c} proposed a solution that can withstand transient (it is self-stabilizing) and Byzantine failures (with a containment radius of 2), so it is also a solution for tolerating transient and crash faults. However, even in the case of crash faults, a containement radius of 2 is also a lower bound~\cite{PS04c} when the system is asynchronous. The same paper \cite{NA02c} shows that global tasks such as mutual exclusion cannot admit a constant radius fault-containing solution when both transient and Byzantine fault are considered. It would be interesting to investigate whether limiting the fault model to transient faults and process crashes permits to break this impossibility result.

\bibliographystyle{plain}
\bibliography{biblio}

\end{document}